\documentclass{article}
\usepackage{graphicx} 
\usepackage{amsthm, amsmath, amsfonts,mathtools, fullpage, graphicx, wrapfig, tikz, caption, subcaption}
\usepackage{natbib}
\usepackage{amssymb}
\usepackage{verbatim}
\usepackage{relsize}
\usepackage{listings}
\usepackage{bbm}
\usepackage{xcolor}
\usepackage{hyperref}
\usepackage{graphicx}
\usepackage{multirow}
\usepackage{authblk}

\usepackage[linesnumbered,ruled,vlined]{algorithm2e}
\newcommand\independent{\protect\mathpalette{\protect\independenT}{\perp}}
\def\independenT#1#2{\mathrel{\rlap{$#1#2$}\mkern2mu{#1#2}}}
\newtheorem{theorem}{Theorem}
\newtheorem{lemma}{Lemma}

\newtheorem{remark}{Remark}
\newcommand{\norm}[1]{\left\lVert#1\right\rVert}

\usepackage{setspace}

 \usepackage[vmargin=1.18in,hmargin=1.1in]{geometry}

\title{Structural Nested Mean Models\\ Under Parallel Trends Assumptions}
\author[1]{Zach Shahn}
\author[2]{Oliver Dukes}
\author[1]{Meghana G. Shamsunder}
\author[3]{David Richardson}
\author[4]{Eric Tchetgen Tchetgen}
\author[5]{James Robins}

\affil[1]{CUNY School of Public Health, New York, NY, USA}
\affil[2]{Ghent University, Ghent,  Belgium}
\affil[3]{University of California Irvine, Irvine, CA, USA}
\affil[4]{University of Pennsylvania, Philadelphia, PA, USA}
\affil[5]{Harvard TH Chan School of Public Health, Boston, MA, USA}

\begin{document}

\maketitle

\begin{abstract}

In this paper, we link and extend two approaches to estimating time-varying treatment effects on repeated continuous outcomes—time-varying \textit{Difference in Differences} (DiD; see \citet{roth2023} and \citet{chaisemartin_review} for reviews) and \textit{Structural Nested Mean Models} (SNMMs; see \citet{snmm_review2014} for a review). In particular, we show that SNMMs, which were previously only known to be nonparametrically identified under a no unobserved confounding assumption, are also identified under a modified version of the parallel trends assumption typically used to justify time-varying DiD methods. Because SNMMs model a broader set of causal estimands, our results allow practitioners of existing time-varying DiD approaches to address additional types of substantive questions under similar assumptions. SNMMs enable estimation of time-varying effect heterogeneity, the lasting effects of a “blip” of treatment at a single time point, effects of sustained interventions (possibly on continuous or multi-dimensional treatments) when treatment repeatedly changes value in the data, controlled direct effects, effects of dynamic treatment strategies that depend on covariate history, and more. Our results also allow analysts who apply SNMMs under the no unobserved confounding assumption to estimate some of the same causal effects under alternative identifying conditions and thus potentially to triangulate evidence. We provide a method for sensitivity analysis to violations of our parallel trends assumption.
We further explain how to estimate optimal treatment regimes via optimal regime Structural Nested Mean Models under parallel trends assumptions plus an assumption that there is no effect modification by unobserved confounders. Finally, we illustrate our methods with real data applications estimating effects of Medicaid expansion on uninsurance rates, effects of floods on flood insurance take-up, and effects of sustained changes in temperature on crop yields.
\end{abstract}


\section{Introduction}

In this paper, we link and extend two approaches to estimating time-varying treatment effects on repeated continuous outcomes—time-varying Difference in Differences (DiD; see \citet{roth2023} and \citet{chaisemartin_review} for reviews) and Structural Nested Mean Models (SNMMs; see \citet{snmm_review2014} for a review). In particular, we show that SNMMs, which were previously only known to be nonparametrically identified under a no unobserved confounding assumption, are also identified under parallel trends assumptions similar to those typically used to justify time-varying DiD methods. Because SNMMs model a broader set of causal estimands, our results allow practitioners of time-varying DiD approaches to address additional types of substantive questions (such as time-varying effect heterogeneity, the lasting effects of a “blip” of treatment at a single time point, effects of sustained interventions on possibly continuous valued treatment variables, effects of dynamic treatment rules that depend on covariate history, and others) under similar assumptions. Our results also allow analysts who apply SNMMs under the no unobserved confounding assumption to estimate some of the same causal effects under alternative identifying conditions and thus potentially to triangulate evidence.

SNMMs \citep{jamie1994,jamie1997,robins2004optimal} are models of time-varying effect heterogeneity in the treated. 
In the so-called `staggered adoption' setting (in which all units start off untreated and then initiate treatment in a staggered fashion) that is the focus of most time-varying DiD literature, SNMMs can model conditional effects of first treatment initiation in the initiators as a function of time-varying covariate history up to the time of initiation. For example, an SNMM could model how the effect of Medicaid expansion varies with unemployment rate prior to expansion. (See our real data application in Section \ref{medicaid}.) Alternative DiD approaches cannot characterize such time-varying effect heterogeneity. Marginal effects of initiation at each time (i.e. not conditional on time-varying covariate history) targeted by standard DiD approaches are also identified given SNMM parameters by averaging the conditional effects. 

When treatment is not an absorbing state and may repeatedly change value over time in the data, investigators may be interested in effects of interventions setting treatment at all future time points, rather than merely at the time of initiation. Under an appropriate parallel trends assumption, SNMMs can model the conditional effects of initiating and \textit{sustaining} a treatment in those who initiated it, even if some initiators later discontinue in the data. Such effects have caused some consternation in the DiD literature (see Section 3.4 of \citet{roth2023}). As an example with a continuous valued treatment, in Section \ref{crops} we use an SNMM to estimate the average effects of sustained changes to `growing degree days' starting in year $m$ on future county-level crop yields, conditional on past growing degree days. 

SNMMs also directly model the (conditional) lasting effects of a final `blip' of treatment. For example, one might employ an SNMM to model the (conditional) lasting effects on student test scores of a single grade of school with a highly rated teacher followed by teachers with average ratings in subsequent grades. In Section \ref{floods}, we estimate the lasting impact of one final local flood compared to no further floods on county flood insurance take-up given flood history. Such contrasts could not be targeted by standard DiD methods under parallel trends assumptions.

If one is willing to assume a parallel trends assumption specific to an arbitrary treatment strategy $g$, we show that the general regime SNMM \citep{robins2004optimal} for strategy $g$ is identified and along with it the expected outcome trajectory had the population followed strategy $g$. In particular, $g$ can even be a dynamic regime in which treatment assignment at each time depends on covariate and treatment history. For example, $g$ might state `administer vasopressors whenever mean arterial blood pressure falls below 65 mmHG'. Effects of such complex strategies cannot be estimated via standard DiD approaches.
 
Another advantage of SNMMs is that they straightforwardly admit multi-dimensional treatments with continuous valued and/or discrete components. The problem of estimating effects of time-varying continuous valued treatments when units do not share a baseline treatment value has not yet been solved in the DiD literature \citep{de2024difference} but is straightforward to solve using SNMMs (as in the growing degree days and crops example in Section \ref{crops}, mentioned above). Moreover, multidimensional treatments enable estimation of controlled direct effects \citep{robins1992identifiability}, such as effects of Medicaid expansion on bankruptcy rates barring concurrent or future minimum wage increases.

 We further show that under the much stronger assumptions of parallel trends under \textit{all} treatment strategies and no effect modification by unobserved confounders (where the second assumption is essentially implied by the first, as we demonstrate in Appendix \ref{appendix_effect_mod}), it is possible to identify the parameters of an optimal regime SNMM \citep{robins2004optimal} and thus estimate optimal dynamic treatment strategies. This last result might be considered a contribution to the reinforcement learning literature, where estimation of optimal regime SNMMs is sometimes called ‘A-learning’ \citep{schulte2014q}.

 We note that throughout we are assuming availability of panel data, i.e. data on the same units over time, with repeated outcome measures. (Units might be geographic and their outcomes composites of individual residents, such as unemployment rate in a county. In this case, our method would not require that the same individual residents are surveyed at each time point to ascertain the geographic unit level outcomes.) We also stress that while many SNMM studies consider settings with an outcome measured only once at the end of follow up, the methods in this paper require repeated outcome measurements. 


The organization of the paper is as follows. In Section \ref{background}, we introduce notation, state assumptions, and describe causal contrasts that can be targeted by additive SNMMs. In Section \ref{section_id_est}, we establish identification and provide estimators for the parameters of an additive SNMM under our parallel trends assumption. In Section \ref{mult_snmm}, we extend our results to SNMMs for multiplicative effects. In Section \ref{section_general}, we discuss estimation of the optimal regime via an optimal regime SNMM under parallel trends assumptions for all regimes plus an additional `no effect modification by unobserved confounders' assumption. In Section \ref{section_real_data}, we present several applications to real data, including effects of Medicaid expansion on uninsurance rates, effects of floods on flood insurance take-up, and effects of temperature on crop yields. In Section \ref{section_sens}, we discuss sensitivity analysis for violations of conditional parallel trends. In Section \ref{section_other_work}, we situate the contributions of this paper within the time-varying DiD literature. We explain why our parallel trends assumption (which was concurrently proposed by \citet{renson2023identifying}), might be more plausible than alternative assumptions common in the DiD literature (e.g. \citet{callaway2021difference}) in the presence of time-varying confounders. In Section \ref{section_conclusion}, we conclude. The theme of the paper is that it can be fruitful to use SNMMs under assumptions similar to those that have traditionally justified DiD.

\section{Notation, SNMMs, and Assumptions}\label{background}
\subsection{Notation}
Suppose we observe a cohort of $N$ subjects indexed by $i \in \{1,\ldots,N\}$. Assume that each subject is observed at regular intervals from baseline
time $0$ through end of follow-up time $K$, and there is no loss to
follow-up. At each time point $m$, the data are collected on $O
_m=(Z_m,Y_m,A_m)$ in that temporal order. $A_m$ denotes the (\textit{possibly
multidimensional with discrete and/or continuous components}) treatment
received at time $m$, $Y_m$ denotes the outcome of interest at time $m$, and 
$Z_m$ denotes a vector of covariates at time $m$ excluding $Y_m$. For
arbitrary time varying variable $X$: we denote by $\bar{X}_m=(X_0,\ldots,X_m)
$ the history of $X$ through time $m$; we denote by $\underline{X}%
_m=(X_m,\ldots,X_K)$ the future of $X$ from time $m$ through time $K$; and, unless otherwise stated,
whenever the negative index $X_{-1}$ appears it denotes the null value. Define $\bar{L}_m$ to be a function of $(\bar{Z}_m,\bar{Y}_{m})$, i.e. a possibly high dimensional summary of the covariate and outcome history through time $m$. We can regard $\bar L_m$ as analyst chosen. All results are stated in terms of an arbitrary $\bar L_m$, but it should be understood that assumptions might hold for some choices and not others depending on the data generating process. We will discuss in Remark \ref{remark_Y_in_L} how inclusion of certain summaries of outcome history in $\bar L_m$ can be problematic.

We adopt the counterfactual framework for time-varying treatments
\citep{robins1986new} which posits that corresponding to each time-varying
treatment regime $\bar{a}_m$, each subject has a counterfactual or potential
outcome $Y_{m+1}(\bar{a}_{m})$ that would have been observed had that
subject received treatment regime $\bar{a}_m$. More generally, let $g$ denote an arbitrary, possibly dynamic treatment strategy, where $%
g\equiv (g_0,g_1,\ldots,g_K)$ is a vector of functions $g_t: (\bar{L}_t,\bar{%
A}_{t-1}) \rightarrow a_t$ that determine treatment values at each time
point given observed history. Let $Y_k(g)$ then denote the counterfactual value of the outcome at time $k$ under $g$. We will adopt the convention that $Y_k(\bar{a}_m,\underline{g}_{m+1})$ denotes the counterfactual outcome under a regime assigning treatments $\bar{a}_m$ through time $m$ and then following $g$ thereafter.

\subsection{SNMMs}

SNMMs are models of time-varying effect heterogeneity. They are concerned with the causal contrasts 
\begin{equation}  \label{blip}
\gamma_{mk}^{g*}(\bar{l}_m,\bar{a}_m)\equiv E[Y_k(\bar{a}_m,\underline{g}_{m+1}) -
Y_k(\bar{a}_{m-1},\underline{g}_m)|\bar{A}_m=\bar{a}_m,\bar{L}_m=\bar{l}_m]
\end{equation}
for all $k>m$. $\gamma_{mk}^{g*}(\bar{l}_m,\bar{a}_m)$ is the average effect at
time $k$ among units with history $(\bar{L}_m=\bar{l}_m,\bar{A}_m=\bar{a}%
_m)$ of receiving treatment $a_m$ at time $m$ and then following $g$ thereafter
compared to following $g$ at time $m$ and thereafter. We will sometimes suppress the dependence of $\gamma_{mk}^{g*}(\bar{l}_m,\bar{a}_m)$ on $g$ and write $\gamma_{mk}^{*}(\bar{l}_m,\bar{a}_m)$ for notational simplicity.

An important special case occurs when $g=\bar{0}$ is the `untreated' regime. (In fact, historically, SNMMs were originally developed for this special case \citep{jamie1994} and were only later extended to general $g$ \citep{robins2004optimal}.) In the $g=\bar{0}$ setting, the contrasts in (\ref{blip}) represent the conditional effects of one last
`blip' of treatment at time $m$ and at level $a_m$. Hence (\ref{blip}) are accordingly often referred to as `blip functions'. For example, in Section \ref{floods}, we consider a blip function representing the lasting conditional effects on flood insurance take-up of a single local flood followed by no further floods compared to no further floods at all, given flood history. 

A \textit{parametric} SNMM imposes functional forms on the blip functions $\gamma_{mk}^{g*}(\bar{l}%
_m,\bar{a}_m)$ for each $k>m$, i.e. 
\begin{equation}  \label{snmm}
\gamma_{mk}^{g*}(\bar{l}_m,\bar{a}_m) = \gamma_{mk}^{g*}(\bar{l}_m,\bar{a}_m;\psi^*),
\end{equation}
where $\psi^*$ is the true value of an unknown finite dimensional parameter vector and $\gamma_{mk}^{g*}(\bar{l}_m,\bar{a}_m;\psi)$ is a known function equal to $0$ whenever $\psi=0$ or $a_m=g(\bar{l}_m,\bar{a}_{m-1})$.

\begin{remark}\label{staggered}\textbf{Effects of Initiation} As we mentioned in the introduction, most applications of time-varying DiD methods have been in staggered adoption settings where all units start off untreated and then begin treatment at different times. Interest centers on the average effect of the observed treatment trajectory compared to never initiating treatment in those that actually initiated treatment at level $a_m$ at time $m$ (and then possibly changed their treatment value afterward under the observational regime). If we let $g=\bar{0}$ and code exposure such that it takes a non-zero value only at the time of first treatment and 0 at all other times, then (\ref{blip}) represents the conditional effects of treatment initiation in the treated. We use this treatment coding when analyzing Medicaid expansion effects in Section \ref{medicaid}. SNMMs for initiation effects have been called `coarse SNMMs' in previous work \citep{robins1998correction,lok2012impact,https://doi.org/10.3982/ECTA17522}.
\end{remark}

\begin{remark}\textbf{Effects of Sustained Interventions}
    A simple but important example of a dynamic regime $g$ is the `sustain previous treatment value' regime
\begin{equation}\label{continue_g}
g_{sus}(\bar{L}_m,\bar{A}_{m-1})=A_{m-1}.
\end{equation}
This regime facilitates estimation of effects of sustained interventions. For example,
\begin{equation*}
\gamma_{mk}^{g_{sus}*}(\bar{l}_m,(\bar{0}_{m-1},1))= E[Y_k(\bar{0}_{m-1},\underline{1})-Y_k(\bar{0})|\bar{A}_m=(\bar{0}_{m-1},1),\bar{L}_m=\bar{l}_m]
\end{equation*}
is the conditional effect of starting and continuing treatment at time $m$ compared to never starting in those who did start at $m$. In Section \ref{crops}, we apply regime (\ref{continue_g}) to study conditional effects of a sustained change in (continuous valued) `growing degree days' on crop yields given growing degree days prior to the change.
\end{remark}

\begin{remark}\label{cde_estimands}\textbf{Controlled Direct Effects} Multidimensional treatments enable the contrasts in (\ref{blip}) to straightforwardly represent controlled direct effects (CDEs) \citep{robins1992identifiability}. Suppose $A_m=(A_{m1},A_{m2})$ is a two dimensional treatment and $g=\bar{0}$ is the untreated regime. 
For example, suppose $A_{m1}$ denotes initial Medicaid expansion (i.e. taking the value $1$ in the first year of expansion and $0$ otherwise, as in Remark \ref{staggered}) and $A_{m2}$ denotes a minimum wage increase in year $m$. $\gamma_{mk}^*(\bar{a}_{m},\bar{l}_m)$ with $a_m=(1,0)$ is then the conditional effect of Medicaid expansion and no future minimum wage increase in year $m$ expanders that did not increase their minimum wage in year $m$.  
\end{remark}
\begin{remark}\label{derived_estimands}\textbf{Derived quantities}  Given knowledge of $\gamma_{mk}^*(\bar{l}_{m},\bar{a}_{m})$,  additional derived quantities of interest are identified under the Consistency assumption (\ref{consistency}) stated in the following subsection \citep{jamie1994,robins2004optimal}. For any subject history $(\bar{L}_m=\bar{l}_m,\bar{A}_m=\bar{a}%
_{m-1})$ of interest, expected conditional counterfactual outcomes
under $\underline{g}_m$, i.e. $E[{Y}_{k}(a_{m-1},%
\underline{g}_m)|\bar{L}_m=\bar{l}_m,\bar{A}_{m-1}=\bar{a}_{m-1}]$ for $k>m$, would also be
identified as $E[Y_{k}-\sum_{j=m}^{k-1}\gamma_{jk}^*(\bar{L}_{j},\bar{A}_{j})|\bar{L}_m=\bar{l}_m,\bar{A}_{m-1}=\bar{a}_{m-1}]$. Quantities that further condition
on treatment at $m$, i.e. $E[Y_{k}(a_{m-1},\underline{g}_m)|\bar{L}%
_m=\bar{l}_m,\bar{A}_{m}=\bar{a}_{m}]$ for $k>m$ would be identified as $E[Y_{k}-\sum_{j=m}^{k-1}\gamma_{jk}^*(\bar{L}_{j},\bar{A}_{j})|\bar{L}_m=\bar{l}_m,\bar{A}_{m}=\bar{a}_{m}]$. Marginalizing over $\bar{L}_m$ then identifies $E[Y_{k}(a_{m-1},\underline{g}_m)|\bar{A}_{m-1}=\bar{a}_{m-1}]$ and $E[Y_{k}(a_{m-1},\underline{g}_m)|\bar{A}_{m}=\bar{a}_{m}]$ for $k>m$. In staggered adoption settings with treatment coded as described in Remark \ref{staggered} and $g=\bar{0}$, $E[\gamma_{mk}^*(\bar{l}_{m},\bar{a}_{m})|\bar{A}_{m}=\bar{a}_{m}]$ are the same causal estimands targeted by standard DiD approaches such as \cite{callaway2021difference} if $a_m=1$. The derived quantities $E[Y_k(g)]$ for each $k>0$, i.e. the expected counterfactual outcome trajectory under no treatment,
would also be identified by $E[Y_k-\sum_{j=0}^{k-1}\gamma_{jk}^*(\bar{L}_j,\bar{A}_j)]$.  
\end{remark}

\subsection{Assumptions}
We make the standard Consistency assumption 
\begin{equation}  \label{consistency}
\textbf{Consistency: } Y_{m}(\bar{A}_{m-1}) = Y_m \text{ } \forall m\leq K
\end{equation}
stating that observed outcomes are equal to counterfactual outcomes
corresponding to observed treatments. Throughout, we will also assume, for regime $g$ of interest,
\begin{equation}  \label{positivity}
\textbf{Positivity: } f_{A_m|\bar{L}_m,\bar{A}_{m-1}}(g(\bar{l}_m,\bar{a}
_{m-1})|\bar{l}_m,\bar{a}%
_{m-1})>0 \text{ whenever } f_{\bar{L}_m,\bar{A}_{m-1}}(\bar{l}_m,\bar{a}%
_{m-1})>0.
\end{equation}
However, we note that under parametric models (\ref{snmm})
positivity is not strictly necessary. 

\citet{jamie1994,jamie1997,robins2004optimal} has shown that $\gamma^{*}=(\gamma^{*}_{01},\ldots,%
\gamma^{*}_{(K-1)K})$, the vector of all blip functions $\gamma_{mk}^*(\bar{l}_m,\bar{a}_m)$ with $m<k\leq K$, is nonparametrically identified and described
how to consistently estimate the parameter $\psi^{*}$ of a parametric SNMM (\ref{snmm}) under the assumption that there are no
unobserved confounders, i.e.
\begin{equation*}
Y_k(\bar{a}_{m-1},\underline{g}_m)\independent A_m|\bar{A}_{m-1}=\bar{a}_{m-1},\bar{L}_m \text{ }\forall k>m,\bar{a}_{m-1}.
\end{equation*}
In this paper, we will instead make the
parallel trends assumption 
\begin{align}
\begin{split}  \label{parallel_trends}
& \textbf{Time-Varying Conditional Parallel Trends: } \\
&E[Y_k(\bar{a}_{m-1},\underline{g}_m) - Y_{k-1}(\bar{a}_{m-1},\underline{g}_m)|%
\bar{A}_m=\bar{a}_m,\bar{L}_m] = \\
&E[Y_k(\bar{a}_{m-1},\underline{g}_m) - Y_{k-1}(\bar{a}_{m-1},\underline{g}_m)|%
\bar{A}_m=(\bar{a}_{m-1},g(\bar{L}_m,\bar{a}_{m-1})),\bar{L}_m] \\
&\forall k>m.
\end{split}%
\end{align}
When $k=m+1$, this assumption reduces to 
\begin{equation}
E[Y_{m+1}(\bar{a}_{m-1},g_m) - Y_{m}|\bar{A}_m=\bar{a}_m,\bar{L}_m] =
E[Y_{m+1}(\bar{a}_{m-1},g_m) - Y_m|\bar{A}_m=(\bar{a}_{m-1},g_m(\bar{L}_m,\bar{a}_{m-1})),\bar{L}_m].
\end{equation}
The time-varying conditional parallel trends
assumption states that, conditional on observed covariate history through
time $m$ and treatment history through time $m-1$, the expected
counterfactual outcome trends under strategy $g$ from time $m$ onwards
do not depend on whether the treatment actually received at time $m$ would have been assigned under $g$. Note that the substantive interpretation of this assumption depends on the particular strategy $g$ for which it is made. Most DiD literature (with the exception of \citet{renson2023identifying}) makes parallel trends assumptions relative to the untreated $g=\bar{0}$ regime. In fact, most literature (again with the exception of \citet{renson2023identifying}) problematically compares counterfactual untreated trends in the treated at time $m$ to trends in the population that is observed to never receive treatment through end of follow up. This conditioning on future treatments makes competing parallel trends assumptions considerably less plausible than our assumption above in the presence of time-varying confounding, as we discuss further when we compare our parallel trends assumption to others in the DiD literature in Section \ref{section_other_work}. In Figure \ref{fig:swig} of that section, we also provide a Single World Intervention Graph (SWIG) \citep{richardson2013single,chernozhukov2024applied} representing sufficient conditions under which our parallel trends assumption would hold.


\begin{remark}\textbf{Conditioning on prior outcomes}\label{remark_Y_in_L}
While we have defined $\bar{L}_m$ as an arbitrary summary of covariate and outcome history, we note that inclusion of past outcomes can be problematic. If the most recent outcome $Y_m$ were included in $\bar L_m$, then when $k=m+1$ the parallel trends assumption would imply that there
is no unobserved confounding (as first pointed out by \citet{laird1983further}), which we do not wish to assume as then the
estimators we introduce would not be needed. This means we cannot adjust for or estimate effect heterogeneity conditional on the most
recent outcome with our methods. In fact, under the SWIG in Figure \ref{fig:swig}, adjustment for any past outcome levels would lead to parallel trends violations due to collider bias, but adjustment for past \textit{increments} would be allowed. Thus, under that SWIG, $\bar L_m$ cannot contain $Y_t$ for any $t\leq m$, but can (and indeed must) contain increments $\bar\Delta_m$ where $\Delta_j\equiv Y_j-Y_{j-1}$.
\end{remark}


\begin{remark}\label{remark_pt_cde}
    \textbf{Parallel trends and CDEs} The parallel trends assumption may become more plausible under a CDE intervention, which is a possible motivation for considering CDEs. Suppose an investigator is interested in effects of Medicaid expansion on some outcome that is impacted both by Medicaid and the minimum wage. Given that states that expand Medicaid in a given year are also more likely to later increase the minimum wage, future minimum wage increases can lead to violations of the parallel trends assumption in an analysis where Medicaid is the sole treatment. However, by estimating the joint effect of Medicaid expansion and no future minimum wage expansion compared to never expanding Medicaid or increasing the minimum wage (i.e. the controlled direct effect of Medicaid expansion setting minimum wage increases to 0), the particular threat to the parallel trends assumption posed by future minimum wage increases is eliminated. 
\end{remark}

\section{\protect\normalsize Identification and g-Estimation of SNMMs
Under Parallel Trends}\label{section_id_est} 

\subsection{\protect\normalsize Nonparametric Identification}

{\normalsize For any $m=0,...,K-1$, let $H_{m}(\gamma)$ denote a vector with components 
\begin{equation}  \label{H}
H_{mk}(\gamma)\equiv Y_k-\sum_{j=m}^{k-1}\gamma_{jk}(\bar{L}_j,\bar{A}_j) 
\end{equation}
}{\normalsize for $k=m+1,...,K$. Here, the $\gamma_{jk}(\bar{L}_j,\bar{A}_j)$ are arbitrary functions of covariate and treatment history not necessarily equal to the true blip functions $\gamma^{*}_{jk}(\bar{L}_j,\bar{A}_j)$. \citet{jamie1994,robins2004optimal} showed that for every $m=0,...,K-1$, $H_{m}(\gamma^*)$ (i.e. $H_{m}(\gamma)$ evaluated at the true blip function) has the
important property: 
\begin{align}
\begin{split}  \label{H_cf}
&E[H_{m}(\gamma^*)|\bar{L}_m,\bar{A}_{m}] = E[\underbar{Y}_{m+1}(\bar{A}_{m-1},%
\underline{g}_m)|\bar{L}_m,\bar{A}_{m}].
\end{split}%
\end{align}
}

{\normalsize 
By (\ref{H_cf}) and the time-varying conditional parallel
trends assumption (\ref{parallel_trends}), it follows that for all $m=0,...,K-1$ and $k=m+1,...,K$,
\begin{align}  
\begin{split}\label{id}
&E[H_{mk}(\gamma^*) - H_{m,k-1}(\gamma^*)|\bar{A}_m=\bar{a}_m,\bar{L}_m] =
E[H_{mk}(\gamma^*) - H_{m,k-1}(\gamma^*)|\bar{A}_m=(\bar{a}_{m-1},g_m(\bar{L}_m,\bar{a}_{m-1})),\bar{L}%
_m].
\end{split}
\end{align}
That is, given the true blip function, the observable quantity $%
H_{mk}(\gamma^*)$ behaves like the counterfactual quantity $Y_k(\bar{A}%
_{m-1},\underline{g}_m)$ in that its conditional trend does not depend on $A_m$%
. We can exploit this property to establish identification and construct multiply robust estimating functions for $\gamma^*$. 

\begin{theorem}\label{standard_id}
Under (\ref{consistency}), (\ref{positivity}), and (\ref%
{parallel_trends}),\newline
(i) $\gamma^{*}$ is identified from the joint distribution of $(\bar{L}%
_{K-1},\bar{A}_{K-1},\bar{Y}_K)$ as the unique solution to 
\begin{equation}  \label{np_est_eq}
E\left\{\sum^{K-1}_{m=0} \sum^{K}_{k=m+1}[s_{mk}(\overline{L}_{m},\bar{A}_{m})-E\{s_{mk}(\overline{L}_{m},\bar{A}_{m})|%
\overline{L}_{m},\bar{A}_{m-1}\}]\{H_{mk}(\gamma) - H_{m,k-1}(\gamma)\}\right\} =0
\end{equation}
where $s_{mk}(\overline{l}_{m},\bar{a}_{m})$ is an arbitrary function of $(\overline{l}_{m},\bar{a}_{m})$ for $m=0,...,K$ and $k=m+1,...,K$.
\newline
(ii) Let $\tilde{E}\{s_{m}(\overline{L}_{m},\overline{A}%
_{m})|\overline{L}_{m},\overline{A}_{m-1}\}$ be a conditional expectation taken with respect to some density/mass function $\tilde{f}(A_m|\overline{L}_{m},\overline{A}_{m-1})$ not necessarily equal to the true density $f(A_m|\overline{L}_{m},\overline{A}_{m-1})$. Similarly, for all $m=0,...,K-1$ and $k=m+1,...,K$, let $\tilde{E}\{H_{mk}(\gamma) - H_{m,k-1}(\gamma)|\overline{L}_{m},\overline{A}_{m-1}\}$ be some function not necessarily equal to $E\{H_{mk}(\gamma) - H_{m,k-1}(\gamma)|\overline{L}_{m},\overline{A}_{m-1}\}$.

Then $\gamma^*$ satisfies 
\begin{align}\label{estimating_equations}
E\bigg(\sum^{K-1}_{m=0} \sum^{K}_{k=m+1}&[s_{mk}(\overline{L}_{m},\bar{A}_{m})-\tilde{E}\{s_{mk}(\overline{L}_{m},\bar{A}_{m})|%
\overline{L}_{m},\bar{A}_{m-1}\}]\nonumber \\&\times [H_{mk}(\gamma^*)-H_{m,k-1}(\gamma^*)-\tilde{E}\{H_{mk}(\gamma^*) - H_{m,k-1}(\gamma^*)|\overline{L}_{m},\overline{A}_{m-1}\}]\bigg) =0
\end{align}
for any arbitrary $s_{mk}(\overline{l}_{m},\bar{a}_{m})$, if for any $m=0,...,K-1$ and $k=m+1,...,K$, either

(a) $\tilde{E}\{H_{mk}(\gamma^*) - H_{m,k-1}(\gamma^*)|\overline{L}_{m},\overline{A}_{m-1}\}=E\{H_{mk}(\gamma^*) - H_{m,k-1}(\gamma^*)|\overline{L}_{m},\overline{A}_{m-1}\}$, or

(b) $\tilde{f}(A_m|\bar{L}_m,\bar{A}_{m-1})=f(A_m|\bar{L}_m,\bar{A}_{m-1})$.\newline
\newline
Thus, (\ref{estimating_equations}) is a multiply robust estimating
function. 
\end{theorem}
The proof can be found in Appendix \ref{appendix_theorem1}.
\subsection{\protect\normalsize Estimation}\label{estimation}

Theorem \ref{standard_id} establishes nonparametric identification of time-varying heterogeneous effects (\ref{blip}) under a parallel trends
assumption conditional on time-varying covariates. In practice, estimation would
proceed by first specifying a model (\ref{snmm}) with finite
dimensional parameter $\psi^{\ast}$.

To construct an estimator of $\psi^*$, we shall need to estimate multiple unknown nuisance functions. In what follows, for any $m=0,...,K$ we define the vector $H^\dagger_{m}(\gamma)$, which has components:
\[H^\dagger_{mk}(\psi)\equiv H_{mk}(\psi) - H_{m,k-1}(\psi).\]} Also, for all $m=0,...,K-1$, let 
\[\nu^*_{m}(\overline{l}_{m},\bar{a}_{m-1})\equiv E\{H^\dagger_{m}(\psi^*)|\bar{L}_m=\bar{l}_m,\bar{A}_{m-1}=\bar{a}_{m-1}\}.\] 
We will assume from here on that either $A_{m}$ is Bernoulli or $s_{mk}(\overline{l}_{m},\bar{a}_{m})$ from the previous section is linear in $a_m$; see Remark \ref{multi_cont} for a discussion of the general case. Then we must estimate  $v^*_{mk}(\bar{L}_{m},\bar{A}_{m-1})$ and $\pi^*_m(\bar{L}_{m},\bar{A}_{m-1}) \equiv E(A_m|\bar{L}_{m},\bar{A}_{m-1})$ (hereafter nuisance functions). We will consider state-of-the-art cross-fit doubly robust machine learning (DR-ML) estimators of $\psi^{\ast}$ in which the nuisance functions are estimated by arbitrary machine learning algorithms chosen by the analyst \citep{chernozhukov2018double,smucler2019unifying}. 

In what follows, $\mathbb{P}_N$ denotes the sample average and $\mathbb{P}$ is the true data generating law. If not stated otherwise, the expectation $E[\cdot]$ is evaluated at $\mathbb{P}$. For any function $f(O)$, we let $\norm{f}$ denote $\{\int f(O)^2 d\mathbb{P}(O)\}^{1/2}$; we also use $\norm{\cdot}_2$ to denote the Euclidean norm. 


Algorithm \ref{algorithm:cfdrml-2t} computes our cross fit estimator $\hat{\psi}^{cf}$, where $s=(s_0,\ldots,s_{K-1})$ is a user chosen vector of conformable vector functions. Note that while we have described the algorithm for two folds, it is straightforward to generalize to greater than two.

\begin{algorithm}
\caption{Implementation of a cross-fit DR-ML estimator}\label{algorithm:cfdrml-2t}
Randomly split the $N$ study subjects into two parts: an
estimation sample of size $n$ and a training (nuisance) sample of size $
n_{tr} = N-n$ with $n/N \approx 1/2$. Without loss of generality we shall
assume that $i = 1,\ldots,n$ corresponds to the estimation sample.\\ 
Applying machine learning methods to the training sample data, construct estimators $\hat{\pi}^{(-1)}$ and $\hat{\nu}^{(-1)}$ of $\pi^*$ and $\nu^*$, where $\pi^*\equiv (\pi^*_0,...,\pi^*_{K-1})$,  $\mu^*\equiv (\mu^*_0,...,\mu^*_{K-1})$ and likewise for $\hat{\pi}^{(-1)}$ and $\hat{\nu}^{(-1)}$.\\
Consider the estimating function
\begin{align*}
U_{m}(O;\psi,s_m,\pi,\nu)=&s_{m}(\overline{L}_{m},\overline{A}_{m-1})\{A_m-\pi_m(\overline{L}_{m},%
\overline{A}_{m-1})\}\{H^\dagger_{m}(\psi)-\nu_{m}(\overline{L}_{m},%
\overline{A}_{m-1})\}
\end{align*}
where $s_{m}(\overline{L}_{m},\overline{A}_{m-1})$, $m=0,...,K$ is an arbitrary $p\times (K-m)$-dimensional function of the covariate and treatment history.
Compute $\hat{\psi}^{(1)}$ from the $n$ subjects in the estimation sample as the (assumed unique) solution to vector estimating
equations \[\mathbbm{P}_N\left\{\sum^{K}_{m=0}U_m(\psi,s_m,\hat{\pi}^{(-1)},\hat{\nu}^{(-1)})\right\}=0.\]\\
Next, compute $\hat{\psi}^{(2)}$ just as $\hat{\psi}^{(1)}$, but with the training and estimation samples
reversed. \\
Finally, the cross fit estimate $\hat{\psi}^{cf}$ is $(\hat{\psi}^{(1)} + \hat{\psi}^{(2)})/2$.
\end{algorithm}

\begin{remark}\textbf{Handling nuisance functions that contain the parameter of interest}
Because step 2 of the cross-fitting procedure estimates conditional expectations of functions of $\psi$, the estimator $\hat{\psi}^{(1)}$ must, in general, be solved iteratively and might be difficult to compute. Remark 20 in Appendix 5 of \citet{liu2021efficient} discusses strategies for reducing the computational burden. If parameters $\psi_{mk}$ and $\psi_{m'k}$ of blip functions $\gamma^g_{mk}$ and $\gamma^g_{m'k}$ are variation independent for $m\neq m'$, then it is possible to backwards recursively estimate $\psi_{mk}^*$ given estimates $\hat{\psi}_{m'k}$ for $m'>m$. If, for each $j$ and $k$, $\gamma_{jk}^*(\bar{l}_j,\bar{a}_{j-1};\psi_{jk})$ is further assumed to be linear in $\psi_{jk}$, i.e. $\gamma_{jk}^*(\bar{L}_j,\bar{A}_{j-1};\psi_{jk})=\psi_{jk}^TR_{jk}$ for a given vector transformation $R_{jk}(\bar{L}_j,\bar{A}_j)$, then it is possible to estimate $E[H_{mk}(\psi^*)-H_{m,k-1}(\psi^*)|\bar{L}_m,\bar{A}_{m-1}]$ noting that
\begin{align*}
&E[H_{mk}(\psi_{mk};\hat{\underline{\psi}}_{m+1,k})-H_{m,k-1}(\psi_{m,k-1};\hat{\underline{\psi}}_{m+1,k-1})|\bar{L}_m,\bar{A}_{m-1}]\\
&=E[(Y_k - \sum_{j=m+1}^k\gamma_{jk}^*(\bar{L}_j,\bar{A}_j;\hat{\psi}_{jk}))-(Y_{k-1} - \sum_{j=m+1}^{k-1}\gamma_{jk-1}(\bar{L}_j,\bar{A}_j;\hat{\psi}_{jk-1}))|\bar{L}_m,\bar{A}_{m-1}]\\
&-\psi_{mk}^TE[R_{mk}|\bar{L}_m,\bar{A}_{m-1}] + \psi_{m,k-1}^TE[R_{m,k-1}|\bar{L}_m,\bar{A}_{m-1}]
\end{align*}
where $\hat{\underline{\psi}}_{m+1,k}$ denotes $(\hat{\psi}_{m+1,k},\ldots,\hat{\psi}_{k,k})$; in practice, population expectations would need to be replaced by estimates. The resulting estimator never requires estimation of a conditional expectation of a function of a yet to be estimated $\psi_{mk}$ and is thus straightforward to compute. Let $\hat{\psi}^{(1)}$ denote the parameter estimate obtained from the cross-fitting procedure when $v_{mk}(\overline{l}_{m},\bar{a}_{m-1};\psi^*)$ is estimated as described above. To improve efficiency, one can then directly apply machine learning methods to construct an estimator of the nuisance pseudo-outcome regression function $\hat{v}_{m}(k,\overline{l}_{m},\overline{a}_{m-1};\hat{\psi}^{(1)})$ with $\hat{\psi}^{(1)}$ plugged in to construct the pseudo-outcomes, then construct a new cross-fit estimator $\tilde{\psi}^{(1)}$ using $\hat{v}_{m}(k,\overline{l}_{m},\overline{a}_{m-1};\hat{\psi}^{(1)})$.  See \citet{lewis2020double} for a related approach to estimation enabling lasso estimation of sparse high dimensional blip function parameters. If there is insufficient data to estimate variation independent blip function parameters at each time point, analysts might consider adding additional parametric restrictions, which we discuss in Remarks \ref{parametric_standard_remark} and \ref{linear_standard_remark} below. 
\end{remark}

The following theorem provides the basis of inference of $\psi^*$
\begin{theorem}\label{theorem_est}
 Suppose that the following conditions hold:
\begin{enumerate}
    \item The map $\psi \mapsto \mathbb{P}\left\{\sum^K_{m=0}U_m(\psi,s_m,\pi,\nu)\right\}$ is differentiable at $\psi^*$ uniformly in $(\pi,\nu)$. 
    \item The matrix 
    \[V(\psi^*,s,\pi^*)\equiv  -E\left\{\sum^K_{m=0} s_{m}(\overline{L}_{m},\overline{A}_{m-1})\{A_m-\pi^*_m(\overline{L}_{m},%
\overline{A}_{m-1})\}\frac{\partial}{\partial \psi}H^\dagger_{m}(\psi)|_{\psi=\psi^*}\right\}.\]
    is nonsingular.
    \item For $m=0,...,K$ and for a selected $p \times (K-m)$ function $s_m(\bar{L}_m,\bar{A}_{m-1})$, the maximum absolute value of any element in $s_m(\bar{L}_m,\bar{A}_{m-1})$ is upper bounded by $C$ with probability one, where $C<\infty$.
    \item $\hat{\psi}^{(s)}-\psi^*=o_{\mathbb{P}}(1)$, $\norm{\hat{\pi}^{(-s)}_{m} - \pi_{m}^*}=o_\mathbb{P}(1)$ and $\norm{\hat{\nu}^{(-s)}_{mk}-\nu_{mk}^*}=o_\mathbb{P}(1)$ for $m=0,...,K$, $k=m+1,...,K$ and $s=1,2$.
    \item For $m=0,...,K-1$, $k=m+1,...,K$ and $s=1,2$,
    \begin{align*}
    \norm{\hat{\pi}^{(-s)}_{m} - \pi_{m}^*}\norm{\hat{\nu}^{(-s)}_{mk}-\nu_{mk}^*}=& o_\mathbb{P}(N^{-1/2}).
    \end{align*}  
    \item For any fixed $(\pi,\nu)$, the class $\left\{\sum^K_{m=0}U_m(\psi,s_m,\pi,\nu):\psi \in \mathbb{R}^p\right\}$ is Donsker in $\psi$.
\end{enumerate}

Then $\hat{\psi}^{cf}$ satisfies the expansion
\begin{equation*}
\sqrt{N}(\hat{\psi}^{cf}-\psi^{*})=V(\psi^*,s,\pi^*)^{-1} \frac{1}{\sqrt{N}}\sum_{i=1}^N\sum_{m=0}^KU_m(O_i;\psi^*,s_m,\pi^*,\nu^*)+o_\mathbb{P}(1).
\end{equation*}
\end{theorem}
Hence we have shown that under fairly standard conditions in semiparametric causal inference, $\hat{\psi}^{cf}$ is an asymptotically linear estimator of $\psi^*$; furthermore, it can be straightforwardly shown that its asymptotic variance is equal to the variance of its influence function.



It follows from (\ref{H_cf}) and Theorems 3 and 4, that we can
also construct regular and asymptotically linear (RAL) plug-in estimators of other quantities of interest using $\gamma(\hat{\psi}^{cf})$. $\mathbb{P}_N[H_{0k}(\hat{\psi}^{cf})]$ is a RAL estimator for $%
E[Y_k(g)]$. $\mathbb{P}_N[H_{mk}(
\hat{\psi}^{cf})]$ is a RAL estimator for $E[Y_k(\bar{A}_{m-1},\underline{g}_m
)]$. $\mathbb{P}_N^{B}[H_{mk}(\hat{\psi}^{cf})]$ is
a RAL estimator for $E[Y_k(\bar{A}_{m-1},\underline{g}_m)|(\bar{A}_{m-1},\bar{L%
}_m)\in B]$ where $\mathbb{P}^{B}_N[\cdot]$ denotes sample average among subjects with $(\bar{A}_{m-1},\bar{L}_m)\in B$ and $B$ an event in the sample space of treatment and covariate history with positive probability. For example, $B$ could be the event that there was no treatment through $m-1$ and the average value of another covariate through time $m$ exceeds some threshold. Of course, we can estimate the effect in the treated for a particular history $(\bar{A}_{m},\bar{L}_m)=(\bar{a}_{m},\bar{l}_m)$ (\ref{blip}) directly by $\gamma_{mk}^*(\bar{%
l}_{m},\bar{a}_m;\hat{\psi}^{cf})$. 

\begin{remark}\textbf{Parametric nuisance models}\label{parametric_standard_remark}
Alternatively, one can specify parametric nuisance models for $\pi^*_m(\bar{L}_m,\bar{A}_{m-1})$ and $v_{m}^*(\bar{L}_m,\bar{A}_{m-1};\psi^*)$, respectively, and forego cross fitting. Under the assumptions of Theorem %
\ref{standard_id}, if at each $m=0,....,K-1$ and $k=m+1,...,K$, either a model for $\pi^*_m(\bar{L}_m,\bar{A}_{m-1})$ or a model for $v_{mk}^*(\bar{L}_m,\bar{A}_{m-1};\psi^*)$ is correctly specified, then the corresponding estimator based on the aforementioned estimating functions is consistent and asymptotically normal under standard regularity conditions \citep{newey1994large}.  Moreover, confidence intervals for $\psi^*$ and other derived quantities of interest can be computed via the nonparametric bootstrap.
\end{remark}

\begin{remark}{\textbf{Non-linear effects}\label{multi_cont}}
In the general case where $A_m$ is not necessarily binary or the treatment effect is non-linear in $A_m$, multiply robust estimation remains possible under a model for the conditional density/mass function $f(A_m|\bar{L}_m,\bar{A}_{m-1})$. For cases where the treatment effect is linear in some vector of transformations of $A_m$, then one can model the conditional mean of each component of the vector.
\end{remark}

\begin{remark}\textbf{Closed form estimation for linear models}\label{linear_standard_remark}
Suppose the blip model is linear in $\psi$ (i.e. $\gamma^*_{mk}(\bar{a}_m,\bar{l}_m;\psi^*) = a_m\psi^{*T}R_{mk}(\bar{a}_{m-1},\bar{l}_m)$ for $R_{mk}(\bar{a}_{m-1},\bar{l}_m)$ some transformation of history through time $m$ the dimension of $\psi^*$) and the nuisance model $v_{mk}(\overline{l}_{m},\bar{a}_{m-1};\phi)$ is specified to be linear in $\phi$ (i.e. $E[H_{mk}(\psi^*)-H_{m,k-1}(\psi^*)|\bar{L}_m,\bar{A}_{m-1}]=\phi^TD_{mk}(\bar{a}_{m-1},\bar{l}_m)$ for $D
_{mk}(\bar{a}_{m-1},\bar{l}_m)$ some transformation of history through time $m$ the dimension of $\phi$). Then the doubly robust estimator $(\hat{\psi},\hat{\phi})$ is available in closed form (without cross-fitting) as
\begin{equation}
(\hat{\psi},\hat{\phi})^T = \biggl[\sum_i\sum_{k>m}(Y_{ik}-Y_{ik-1})\begin{pmatrix}s_{im}X_{im}\\D_{imk}\end{pmatrix}\biggr]\biggl[\sum_i\sum_{k>m}(V_{imk}-V_{imk-1},D_{imk})\begin{pmatrix}s_{im}X_{im}\\D_{imk}\end{pmatrix}\biggr]^{-1},
\end{equation}
where $V_{imk}=\sum_{j=m}^{k-1} A_{ij}R_{jk}$, $X_{im} = A_{im}-\hat{\pi}_{im}(\bar{L}_{im},\bar{A}_{im-1})$, and $s_{im}$ is the usual index function with dimension equal to the dimension of $\psi$. 
\end{remark}

\section{\protect\normalsize Multiplicative SNMMs}\label{mult_snmm}
{\normalsize The SNMM framework also readily handles multiplicative effects
when the parallel trends assumption is assumed to hold on the additive
scale, a scenario that has been discussed in the DiD literature \citep{ciani2019dif}. Define the
multiplicative causal contrasts 
\begin{equation}  \label{mult_blip}
e^{\gamma^{g\times*}_{mk}(\bar{l}_m,\bar{a}_m)}\equiv \frac{E[Y_k(\bar{a}_m,%
\underline{g}_m)|\bar{A}_m=\bar{a}_m,\bar{L}_m=\bar{l}_m]}{E[Y_k(\bar{a}_{m-1},%
\underline{g}_m)|\bar{A}_{m}=\bar{a}_m,\bar{L}_m=\bar{l}_m]}
\end{equation}
for $k>m$. $\gamma^{g\times*}_{mk}(\bar{a}_m,\bar{l}_m)$ is the average
multiplicative effect at time $k$ among units with history $(\bar{A}_m=%
\bar{a}_m,\bar{L}_m=\bar{l}_m)$ of receiving treatment $a_m$ at time $m$ and
then following $g$ thereafter compared to following $g$ at time $m$ and
thereafter.}

{\normalsize A parametric multiplicative SNMM imposes functional forms on the
multiplicative blip functions $\gamma^{g\times*}_{mk}(\bar{A}_m,\bar{L}_m)$
for each $k> m$, i.e. 
\begin{equation}  \label{mult_snmm_eq}
e^{\gamma^{g\times*}_{mk}(\bar{a}_m,\bar{l}_m)} = e^{\gamma^{g\times}_{mk}(%
\bar{a}_m,\bar{l}_m;\psi_{\times}^*)}
\end{equation}
where $\psi_{\times}^*$ is an unknown parameter vector and $%
\gamma^{g\times}_{mk}(\bar{a}_m,\bar{l}_m;\psi_{\times})$ is a known function. }

{\normalsize We make the same parallel trends assumption (\ref%
{parallel_trends}) as in the additive setting.  Let 
\begin{equation}  \label{H_mult}
H^{\times}_{mk}(\gamma)\equiv Y_k exp\{-\sum_{j=m}^{k-1}\gamma_{jk}(%
\bar{A}_j,\bar{L}_j)\} \text{ for } k>m \text{ and } H^{\times}_{tt}\equiv
Y_t.
\end{equation}\citet{jamie1994,robins2004optimal} showed that $H_{mk}^{\times}(\gamma^{g\times*})$ evaluated at the true blip function
has the following important properties for all $k>m$: 
\begin{align}
\begin{split}  \label{mult_H_cf}
E[H_{mk}^{\times}(\gamma^{g\times*})|\bar{L}_m,\bar{A}_{m}] = E[Y_k(\bar{A}%
_{m-1},\underline{g}_m)|\bar{L}_m,\bar{A}_{m}]
\end{split}\\
\begin{split}
E[H_{0k}^{\times}(\gamma^{g\times*})] = E[Y_k(g)].
\end{split}%
\end{align}
}

{\normalsize By (\ref{mult_H_cf}) and the time-varying conditional parallel
trends assumption (\ref{parallel_trends}), it follows that 
\begin{align}
\begin{split}  \label{mult_id}
&E[H^{\times}_{mk}(\gamma^{g\times*}) - H^{\times}_{m,k-1}(\gamma^{g\times*})|%
\bar{A}_m=\bar{a}_m,\bar{L}_m] = \\
& E[H^{\times}_{mk}(\gamma^{g\times*}) - H^{\times}_{m,k-1}(\gamma^{g\times*})|%
\bar{A}_m=(\bar{a}_{m-1},g(\bar{L}_m,\bar{a}_{m-1})),\bar{L}_m] \\
& \forall k>m.
\end{split}%
\end{align}
}

{\normalsize That is, given the true multiplicative blip function, the
quantity $H^{\times}_{mk}(\gamma^{g\times*})$ behaves like the
counterfactual quantity $Y_k(\bar{A}_{m-1},\underline{g}_m)$ in that its
conditional expected trend does not depend on $A_m$. We can again exploit this
property to identify and construct doubly robust estimating equations for $%
\gamma^{g\times*}$ and various derivative quantities of interest.
Identification and estimation theorems and proofs are identical to Section (\ref{section_id_est}) with $H_{mk}^{\times}$ in place of $H_{mk}$, $\gamma^{g\times}$
in place of $\gamma^g$, $\psi_{\times}$ in place of $\psi$, and (\ref%
{mult_H_cf}) in place of (\ref{H_cf}). }

\section{Optimal Regime SNMMs}\label{section_general}

In this section, we assume that parallel trends (\ref{parallel_trends}) holds for all  $g\in\mathcal{G}$, for $\mathcal{G}$ the set of all treatment rules. Let $\bar{U}_K$ denote a time-varying unobserved confounder, possibly multivariate and containing continuous and/or discrete components, that together with observed covariates is sufficient to adjust for all confounding, i.e. 
\begin{equation}\label{alpha_seq_ex}
U-\textbf{Sequential Exchangeability: } A_m \independent \underline{Y}_{m+1}(\bar{a}_{K-1})|\bar{L}_m,\bar{A}_{m-1}=\bar{a}_{m-1},\bar{U}_m \text{ } \forall\text{ } m< K,\bar{a}_{K-1}.
\end{equation}
To identify the optimal dynamic treatment regime, we need the additional strong assumption that $U$ is not an effect modifier, which we formalize below.

\textbf{No Additive Effect Modification by U:} For all regimes $g\in\mathcal{G}$,\begin{align}
\begin{split}\label{no_mod}
&\gamma_{mk}^{g}(\bar{l}_m,\bar{a}_m)=\gamma_{mk}^{g}(\bar{l}_m,\bar{a}_m,\bar{u}_m)\equiv\\
&E[Y_k(\bar{a}_m,\underline{g}_{m+1}) - Y_k(\bar{a}_{m-1},\underline{g}_{m})|\bar{A}_m=\bar{a}_m,\bar{L}_m=\bar{l}_m,\bar{U}_m=\bar{u}_m].
\end{split}
\end{align}
In Appendix \ref{appendix_effect_mod}, we show that (\ref{no_mod}) practically, though not strictly mathematically, follows from (\ref{parallel_trends}) holding for all $g\in\mathcal{G}$.

Suppose we want to maximize the expectation of some utility $Y=\sum_{k=0}^K \tau_k Y_k$ which is a weighted sum of the outcomes we observe at each time step with weights $\tau_k$. Let $Y(g)=\sum_{k=0}^K \tau_kY_k(g)$ denote the counterfactual value of the utility under treatment regime $g$. We want to find $g^{opt}=argmax_g E[Y(g)]$. Under (\ref{alpha_seq_ex}) and (\ref{no_mod}), it follows from results in \citet{robins2004optimal} that $g^{opt}$ is given by the backward recursion: 
\small
\begin{align*}
\begin{split}
&g^{opt}_{K-1}(\bar{L}_{K-1},\bar{A}_{K-2})= argmax_{a_{K-1}} E[Y_{K}(\bar{A}_{K-2},a_{K-1})|\bar{L}_{K-1},\bar{A}_{K-2}]\\
&g^{opt}_{m-1}(\bar{L}_{m-1},\bar{A}_{m-2})= argmax_{a_{m-1}} \sum_{j>m} \tau_jE[Y_{j}(\bar{A}_{m-2},a_{m-1},\underline{g}^{opt}_{m})|\bar{L}_{m-1},\bar{A}_{m-2}].
\end{split}
\end{align*}
That is, the optimal treatment rule at each time step is the rule that maximizes the weighted sum of expected future counterfactual outcomes assuming that the optimal treatment rule is followed at all future time steps. Define the optimal regime blip function
\begin{equation}\label{opt_blip}
\gamma_{mk}^{opt*}(\bar{l}_m,\bar{a}_m)\equiv E[Y_k(\bar{a}_m,\underline{g}^{opt}_{m+1})-Y_k(\bar{a}_{m-1},0,\underline{g}^{opt}_{m+1})|\bar{A}_m=\bar{a}_m,\bar{L}_m=\bar{l}_m].
\end{equation}
This is the conditional effect of treatment level $a_m$ followed by the optimal regime compared to treatment level $0$ followed by the optimal regime in those receiving treatment level $a_m$. Under the assumption that unobserved confounders are not effect modifiers, (\ref{opt_blip}) does not depend on $A_m$, i.e.
\begin{equation}\label{opt_blip_no_am}
\gamma_{mk}^{opt*}(\bar{l}_m,\bar{a}_m)=E[Y_k(\bar{a}_m,\underline{g}^{opt}_{m+1})-Y_k(\bar{a}_{m-1},0,\underline{g}^{opt}_{m+1})|\bar{A}_{m-1}=\bar{a}_{m-1},\bar{L}_m=\bar{l}_m].
\end{equation}
In terms of this blip function, then,
\begin{equation}\label{g_opt_blip}
g^{opt}_m(\bar{L}_m,\bar{A}_{m-1}) = argmax_{a_m} \sum_{j> m} \tau_j \gamma_{mj}^{g^{opt}*}(\bar{L}_m,\bar{A}_{m-1},a_m). 
\end{equation}
Under parametric model $\gamma_{mj}^{g^{opt}*}(\bar{L}_m,\bar{A}_{m-1},a_m) = \gamma_{mj}^{g^{opt}*}(\bar{L}_m,\bar{A}_{m-1},a_m;\psi_{g^{opt}}^*)$, we can estimate $\psi_{g^{opt}}^*$ via g-estimation as follows. First, define
\begin{align*}
\begin{split}
H_{mk}^{opt}(\psi_{g^{opt}}) \equiv Y_k + \sum_{j=m}^k&\{\gamma_{jk}^{g^{opt}*}(\bar{L}_j,\bar{A}_{j-1},argmax_{a_j}\sum_{r=j}^K\tau_r\gamma_{jr}^{g^{opt}*}(\bar{L}_j,\bar{A}_{j-1},a_j;\psi_{g^{opt}});\psi_{g^{opt}}) -\\
&\gamma_{jk}^{g^{opt}*}(\bar{L}_j,\bar{A}_{j-1},A_j;\psi_{g^{opt}})\}.
\end{split}
\end{align*}
Given (\ref{opt_blip_no_am}), Lemma 3.2 of \citep{robins2004optimal} implies that, evaluating at the true parameter value,
\begin{equation}
    E[H_{mk}^{opt}(\gamma^{g^{opt}*})|\bar{L}_m,\bar{A}_m] = E[Y_k(\bar{a}_{m-1},\underline{g}^{opt})|\bar{L}_m,\bar{A}_m].
\end{equation}
Intuitively, this is because the effect of observed treatment compared to 0 treatment (followed by the optimal regime) is subtracted off, then the effect of optimal treatment relative to 0 treatment (again followed by the optimal regime) is added back on. Once again, we have the crucial property that $H_{mk}^{opt}(\psi_{g^{opt}}^*)$ behaves in conditional expectation like the counterfactual $Y_k(\bar{a}_{m-1},\underline{g}^{opt})$. Therefore, in particular, $H_{mk}^{opt}(\psi_{g^{opt}}^*)$ satisfies parallel trends:
\begin{align*}
    &E[H_{mk}^{opt}(\psi_{g^{opt}}^*)-H_{mk-1}^{opt}(\psi_{g^{opt}}^*)|\bar{A}_m=\bar{a}_m,\bar{L}_m]=\\
    &E[H_{mk}^{opt}(\psi_{g^{opt}}^*)-H_{mk-1}^{opt}(\psi_{g^{opt}}^*)|\bar{A}_{m-1}=\bar{a}_{m-1},\bar{L}_m].
\end{align*}
Thus, again applying the logic of g-estimation, we can construct consistent and asymptotically normal estimators of $\psi_{g^{opt}}^*$ as in Section \ref{estimation}, with $H_{mk}^{opt}(\psi_{g^{opt}})$ in place of $H_{mk}(\psi_g)$. We can further consistently estimate the optimal treatment rule as 
\begin{equation}
\hat{g}^{opt}(\bar{L}_m,\bar{A}_{m-1}) =  argmax_{a_m}\sum_{j\geq m}^K \tau_j \gamma_{mj}^{g^{opt}}(\bar{L}_m,\bar{A}_{m-1},a_m;\hat{\psi}_{g^{opt}}).
\end{equation} 
We can also consistently estimate the expected value of the counterfactual utility $E[Y(g^{opt})]$ under the optimal regime by \sloppy$\hat{E}[Y(g^{opt})]=\sum_{k=0}^K \tau_k \mathbb{P}[H_{0k}^{opt}(\hat{\psi}_{g^{opt}})]$.

\section{Real Data Applications}\label{section_real_data}
We illustrate our approach with applications to real data. We first fit an SNMM to model conditional effects of Medicaid expansion on county level uninsurance rates given past unemployment rates. This is a staggered adoption application highlighting the ability of SNMMs to characterize time-varying effect heterogeneity. We also fit an SNMM to model effects of floods on flood insurance take-up using data previously analyzed by \citet{gallagher2014learning}. In this application, the blip function itself is directly of interest, as we estimate lasting effects of a `blip' of exposure, i.e. a flood. Finally, we use an SNMM to estimate the effects of sustained changes in `growing degree days' on crop profitability at the county level using data previously analyzed by \citet{deschenes2012economic}. This example highlights the ability of SNMMs to estimate effects of sustained interventions and effects of continuous valued treatments when there is no shared baseline. We specified parametric nuisance models for all of these applications to illustrate simple implementations. Code and data for these analyses can be found at https://github.com/zshahn/did\_snmm.

\subsection{Impact of Medicaid Expansion on Uninsurance Rates }\label{medicaid}
One of the Affordable Care Act's primary goals was to provide increased insurance coverage through the expansion of Medicaid eligibility. Originally a nationwide policy, the choice to expand Medicaid using federal subsidies was ultimately left to individual states, leading to its staggered adoption beginning in 2014. We estimate how the effect of Medicaid expansion on county level uninsurance rates depends on county level unemployment rate and population. Table \ref{tab:bank} summarizes treatment timing over the period 2014-2019 considered in the study. 
\begin{table}[h!]
\begin{center}
\begin{tabular}{|c|c|} 
 \hline
 Year & $\#$ Counties in States First Expanding Medicaid\\ \hline
 2014 & 1216 \\ \hline
 2015 & 113 \\ \hline 
 2016 & 174 \\ \hline
 2017 & 64 \\ \hline
 2018 & 0 \\ \hline
 2019 & 151 \\ \hline
 Never & 1421 \\
 \hline
\end{tabular}
 \caption{Staggered Medicaid adoption by county}
 \label{tab:bank}
\end{center}
\end{table}
Because this is a staggered adoption setting, we coded treatment as in Remark \ref{staggered}, with $A_m=1$ only if a county's state first expanded Medicaid in year $m$ and $A_m=0$ otherwise. We also took $g=\bar{0}$ to be the untreated regime.

We specified the SNMM  
\begin{equation*}
    E[Y_k - Y_k(\bar{0})|\bar{A}_{m}=(\bar{0}_{m-1},1),\bar{L}_m]=\psi_{0mk} + \psi_{pop} L^{pop}_m+\psi_{unemp}L_m^{unemp},
\end{equation*}
 which posits that the effect varies flexibly with time of expansion and time since expansion through parameters $\psi_{0mk}$ and linearly with the previous year's unemployment rate and 2013 log population through parameters $\psi_{pop}$ and $\psi_{unemp}$. The parallel trends assumption that $E[Y_k(\bar{0})-Y_{k-1}(\bar{0})|\bar{A}_{m}=(\bar{0}_{m-1},1),\bar{L}_m]=E[Y_k(\bar{0})-Y_{k-1}(\bar{0})|\bar{A}_m=\bar{0},\bar{L}_m]$ for all $m$ and $k$ states that future expected uninsurance rate trends under no expansion are similar for counties with similar covariate history regardless of whether their state expanded in year $m$. We specified nuisance models $logit^{-1}(Pr(A_m=1|\bar{A}_{m-1}=0,\bar{L}_m))=\beta_{0m} + \beta_{1m} L^{unemp}_m + \beta_{2m} L^{pop}_m +\beta_{3m}L^{unemp}_mL^{pop}_m+\beta_4(L^{unemp}_m)^2+\beta_5(L^{pop}_m)^2$ and $E[H_{mk}(\psi)-H_{m,k-1}(\psi)|\bar{A}_{m-1}=0,\bar{L}_m]=\lambda_{0mk} + \lambda_{1mk}^T L_m$. We obtained point estimates for the SNMM parameters using the closed form linear estimator from Remark \ref{linear_standard_remark} and estimated standard errors via bootstrap.

Figure \ref{bank_figure1} displays the expected time-varying heterogeneous effects of Medicaid expansion on uninsurance rates for expanding counties of median log population (10.2) at the first quartile versus the third quartile unemployment rate. Expansion was estimated to be more impactful in counties with higher unemployment rates. The estimated positive coefficient of log population $\hat{\psi}_{pop}$ = 0.80 (95\% CI =[0.56, 1.04]) suggests that the impact of Medicaid expansion on uninsurance was greater (i.e. `more negative') in areas with lower population.

\begin{figure}[h]
\centering
\includegraphics[scale=.75]{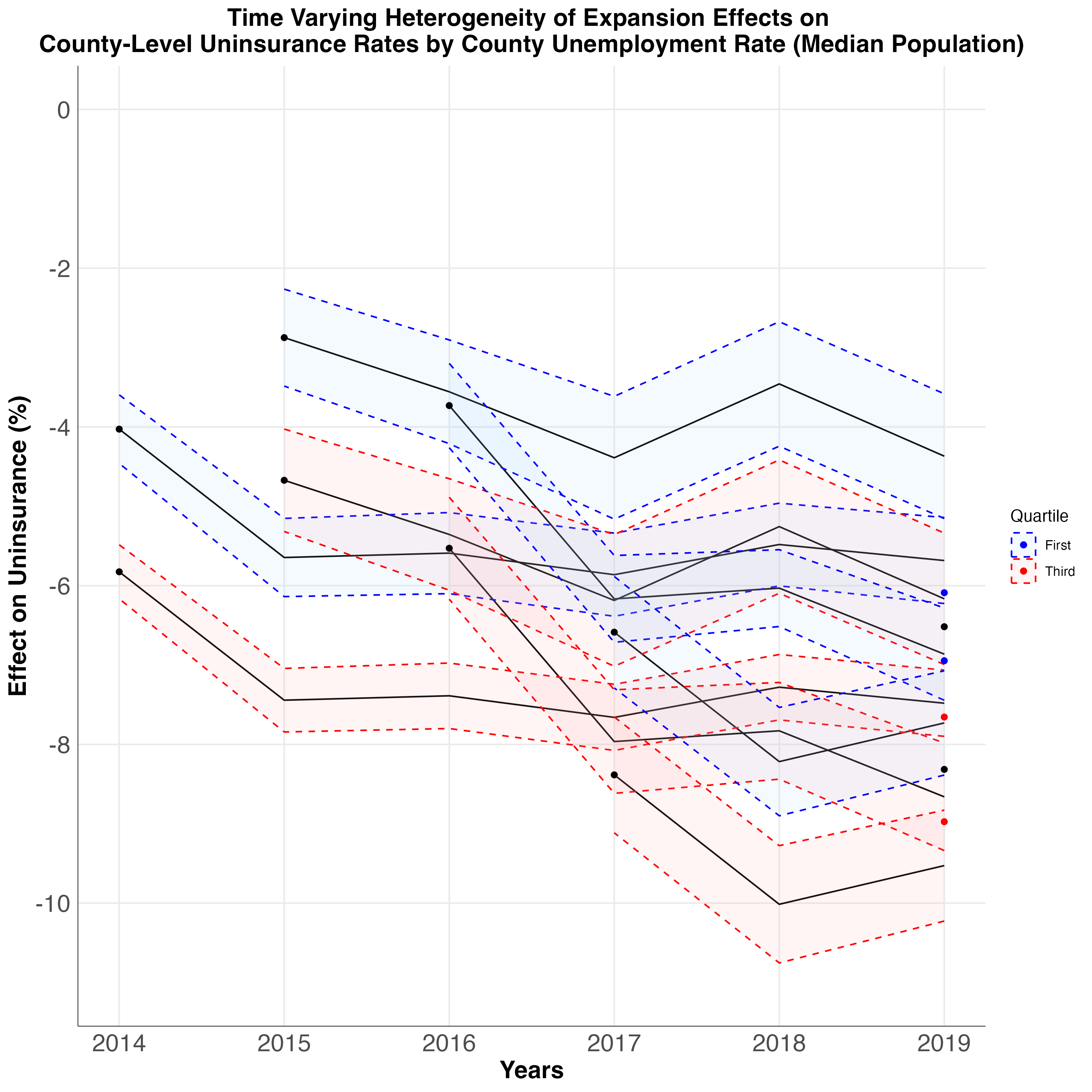}  
\caption{Estimated expected heterogeneous effects (and 95\% confidence intervals) of Medicaid expansion on uninsurance rates in an expanding median population county. Each line represents expected ongoing effects of an expansion at the year the line originates in counties with pre-expansion unemployment rate corresponding to the color of the line. Red lines (effects of expansions in higher unemployment counties)  indicate greater uninsurance reduction effects than blue lines (effects of expansions in lower unemployment counties) originating in the same year. The graph also shows how estimated effects vary by year of expansion and years since expansion. }
\label{bank_figure1}
\end{figure}


\subsection{Impact of Floods on Flood Insurance}\label{floods}
\citet{gallagher2014learning} used a fixed effects regression model to look at effects of floods on flood insurance coverage at the county level. He argued that each county's flood risk is constant over time. We fit a SNMM under the assumption of parallel trends in insurance coverage absent future floods in counties with similar flood history from 1958. We specified a parametric linear blip model $\gamma_{mk}^*(\bar{l}_m,\bar{a}_m)=a_m(1,m-1980,k-m,(k-m)^2,rate_{m-1})^T\psi$ for $g=\bar{0}$, where $rate_{m-1}$ denotes the county's proportion of flood years since 1958. We specified nuisance models $E[A_m|\bar{A}_{m-1},\bar{L}_m]=\beta_{m0} + \beta_{m1} rate_{m-1}$ and $E[H_{mk}(\psi)-H_{m,k-1}(\psi)|\bar{A}_{m-1},\bar{L}_m]=\lambda_{mk0} + \lambda_{mk1} rate_{m-1}$. We obtained blip model parameter estimates via the closed form linear estimators of Remark \ref{linear_standard_remark} and estimated standard errors via bootstrap. Each line in Figure \ref{figure_flood} depicts the estimated effect of a flood occurring at its leftmost time point followed by no further floods on flood insurance uptake over the subsequent 15 years for a county with the median historical flood rate, i.e. $\gamma_{mk}^*(rate_{m-1}=rate_{median},a_m=1;\hat{\psi})$ for $k\in m,\ldots,m+15$. These quantities were directly extracted from our blip function estimates. We see that there is an initial surge in uptake followed by a steep decline, and the estimated initial surge is larger for more recent floods. We did not find statistically significant effect heterogeneity as a function of historical county flood rate. Gallagher (2014) obtained qualitatively similar results and argued that most of the decline in the effect of a flood is due to residents forgetting about it as opposed to migration. It might be interesting to explore other blip model specifications, perhaps conditioning on further aspects of flood history such as years since previous flood or on average flood insurance premiums in the area.

\begin{figure}[h]
\centering
\includegraphics[scale=.5]{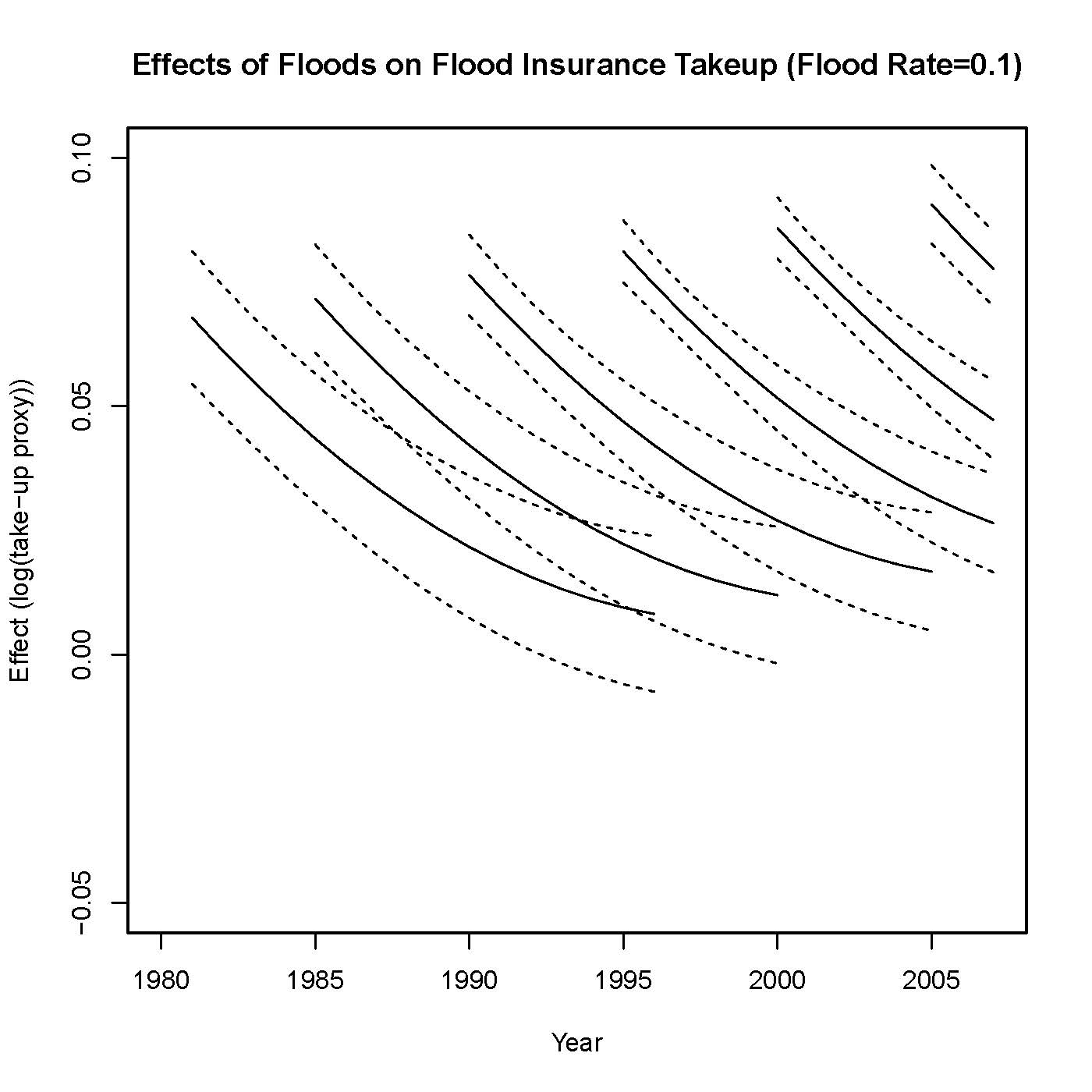}  
\caption{Each line depicts the estimated effect of a flood occurring at its leftmost time point followed by no further floods on flood insurance uptake over the subsequent 15 years for a county with the median historical flood rate, i.e. $\gamma_{mk}^*(rate_{m-1}=rate_{median},a_m=1;\hat{\psi})$ for $k\in m,\ldots,m+15$. These quantities were directly extracted from our blip function estimates.}
\label{figure_flood}
\end{figure}

\subsection{Effects of Sustained Temperature Changes on Crop Profitability}\label{crops}

\citet{deschenes2012economic} analyzed the relationship between growing degree days (GDD) and crop profits per acre at the county level using two way fixed effects regression. (GDD is essentially a composite measure of temperature, truncated at an upper limit each day over the growing season.) They used data from 2,342 counties collected every five years from 1987 through 2002. Two way fixed effects is now known to impose unwanted effect homogeneity assumptions \citep{de2020two}, but \citet{de2024continuous} point out that parallel trends assumptions from the existing DiD literature do not easily accommodate the continuous treatment in this application because counties do not share a baseline level. Furthermore, the exposure changes value repeatedly, another challenge in the DiD literature. (See Figure \ref{crop_figure} for illustration.) We reanalyze the data to estimate effects of a sustained change in GDD on crop profits given history of growing degree days using an SNMM under our conditional parallel trends assumption. 

\begin{figure}[h]
\centering
\includegraphics[scale=.65]{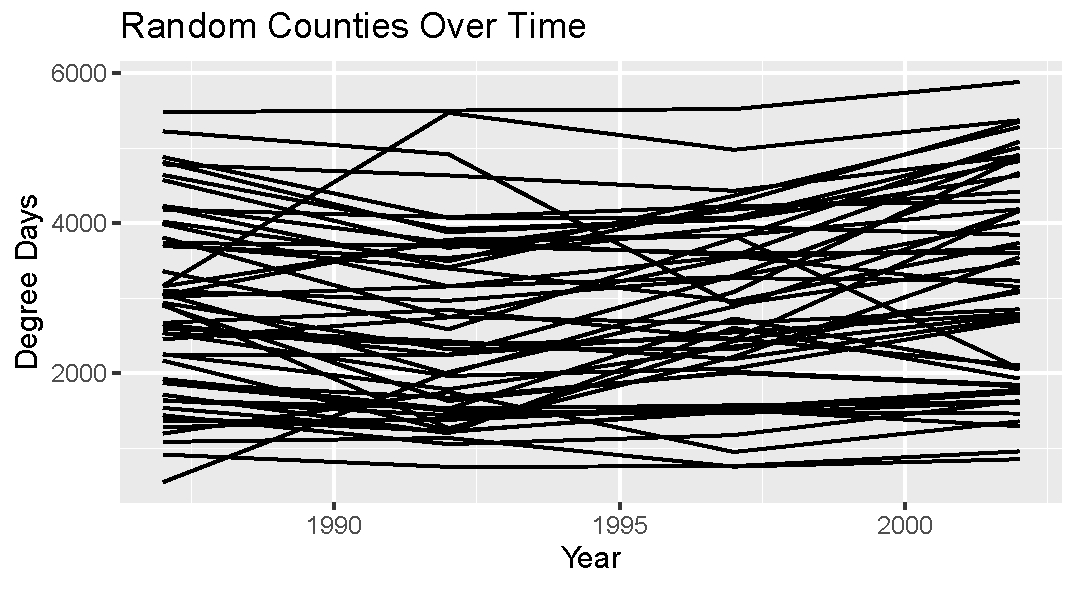}  
\caption{GDD trajectories for 100 random counties from 1987 through 2002. There is not a shared baseline level and values jump repeatedly over time.}
\label{crop_figure}
\end{figure}

Letting $A_m$ denote GDD at measurement $m$. We specify a parametric general regime SNMM
\begin{align}
\begin{split}\label{crop_snmm}
    \gamma_{mk}^*(\bar{a}_m)=E[Y_k(\bar{a}_m,\underline{g})-Y_k(\bar{a}_{m-1},\underline{g})|\bar{A}_m=\bar{a}_m]\\
    =\psi_1a_m+\psi_2a_ma_{m-1}+\psi_3a_m(k-m)+\psi_4a_ma_{m-1}(k-m)+\psi_5a_mm
\end{split}
\end{align}
for $g$ the `keep treatment fixed' regime (\ref{continue_g}).
 (\ref{crop_snmm}) is a model for the effect of a sustained change in GDD to level $a_m$ starting in year $m$ on crop profits in year $k$ compared to permanently fixing GDD at $a_{m-1}$, given the history of GDD through time $m$. The parametric model allows this effect to depend on the new GDD level $a_m$, the prior GDD level $a_{m-1}$, the time $k-m$ since the change, and the time $m$ of the change. It imposes linearity assumptions and the constraint that the effect of the GDD change does not depend on GDD history beyond its prior level. Of course, alternative parametric models could be specified, and we note that the model specification can become arbitrarily flexible as sample size increases. 

We make the parallel trends assumption (\ref{parallel_trends}) that for each time point $m$, in counties with similar GDD histories through $m-1$, future counterfactual crop profit trends fixing GDD at its time $m-1$ level are mean independent of actual GDD at $m$. We specified parametric nuisance models $E[A_m|\bar{A}_{m-1}]=\beta_m^T(1,A_0,A_{m-1},A_0A_{m-1})$ and $E[H_{mk}(\psi)-H_{m,k-1}(\psi)|\bar{A}_{m-1}]=\lambda_{mk0}+\lambda_1A_0+\lambda_2A_{m-1}+\lambda_3A_0A_{m-1}$. We then estimated the blip model parameters using the closed form estimator from Remark \ref{linear_standard_remark} and estimated standard errors by bootstrap. 

Under our assumptions, we can directly query our estimated SNMM to obtain effect estimates of interest. Table \ref{tab:crop_results} contains estimated effects of a sustained increase of 1000 GDD in 1992 compared to keeping GDD at 1987 levels on crop profitability in 1997 and 2002 in counties with various levels of 1987 GDD. Effects were estimated to be larger when the increase was sustained for longer and in counties with higher pre-increase GDD levels.

\begin{table}[h]
    \centering
    \begin{tabular}{|c|c|c|}
        \hline
        \multirow{2}{*}{$GDD_{1987}$} & \multicolumn{2}{c|}{$\gamma_{1992,k}(a_{1992}=GDD_{1987}+1000,a_{1987}=GDD_{1987};\hat{\psi})$} \\
        \cline{2-3}
                                   & k=1997 & k=2002 \\
        \hline
        1000 & 6.1 (-4.9, 17.2)
 & 3.7 (-9.8, 17.2)
 \\
        2000 & 6.3 (-0.7, 13.2)
 & 7.6 (-1.4, 16.5)
 \\
        3000 & 6.4 (2.0, 10.8)
 & 11.4 (4.6, 18.2)
 \\
        4000 & 6.5 (0.4, 12.6)
 & 15.3 (6.4, 24.1)\\
        \hline
    \end{tabular}
    \caption{SNMM (\ref{crop_snmm}) estimated effects of a sustained increase of 1000 GDD from 1987 to 1992 on crop profits per acre in 1997 and 2002 for various baseline GDD levels in 1987.}
    \label{tab:crop_results}
\end{table}

\section{Sensitivity Analysis}\label{section_sens}
Conditional parallel trends assumptions are strong and untestable, and sensitivity analysis for violations of the parallel trends assumption is therefore desirable. We adapt the approach to sensitivity analysis for unobserved confounding in SNMMs of \citet{robins2000sensitivity} and \citet{robins2004optimal} to sensitivity analysis for non-parallel trends. We describe a general class of bias functions characterizing deviations from parallel trends given covariate history. For any particular bias function from this class, we provide a corresponding unbiased estimate of SNMM parameters assuming that the bias function is correctly specified. An analyst can then execute a sensitivity analysis by specifying a plausible range of bias functions (e.g. a grid of parameters covering a plausible range within a parametric subclass of bias functions) and producing the corresponding range of plausible effect estimates. This approach to sensitivity analysis is complementary to that developed by \citet{rambachan2023more}, as it allows deviations to depend on time-varying covariates and allows for sensitivity analysis of all SNMM parameters (e.g. those characterizing effect heterogeneity) and derived quantities. 

Define
\begin{align}
\begin{split}\label{bias_function}
&c_{mk}^g(\bar{l}_m,
\bar{a}_m)=\\
&E[Y_k(\bar{a}_{m-1},\underline{g}_m)-Y_{k-1}(\bar{a}_{m-1},\underline{g}_m)|\bar{L}_m=\bar{l}_m,\bar{A}_m=\bar{a}_m] - \\
&E[Y_k(\bar{a}_{m-1},\underline{g}_m)-Y_{k-1}(\bar{a}_{m-1},\underline{g}_m)|\bar{L}_m=\bar{l}_m,\bar{A}_m=(\bar{a}_{m-1}, g(\bar{l}_m,\bar{a}_{m-1}))].
\end{split}
\end{align}
$c_{mk}^g(\bar{l}_m,\bar{A}_m)$ characterizes the magnitude of deviation from the parallel trends assumption (\ref{parallel_trends}). It is a general function in that it allows deviations to depend on treatment and covariate history, treatment time $m$, and outcome time $k$. If parallel trends holds or $a_m=g(\bar{l}_m,\bar{a}_{m-1})$, then $c^g_{mk}(\bar{l}_m,\bar{a}_m)=0$.

Given a bias function (\ref{bias_function}), define the bias adjusted version of the `blipped down' quantity (\ref{H}) as
\begin{equation}\label{H_bias}
H_{mk}^{a}(\gamma^g)\equiv H_{mk}(\gamma^g)-c^g_{mk}(\bar{L}_m,\bar{A}_m),
\end{equation} 
where $H_{mk}$ is defined in (\ref{H}). 

\begin{lemma}\label{sens_lemma}
If (\ref{bias_function}) is correctly specified, then 
\begin{align}
\begin{split}\label{H_id_bias}
&E[H_{mk}^{a}(\gamma^{g*})-H_{mk-1}(\gamma^{g*})|\bar{L}_m,\bar{A}_m]\\
&=E[H_{mk}^{a}(\gamma^{g*})-H_{mk-1}(\gamma^{g*})|\bar{L}_m,\bar{A}_{m-1}].
\end{split}
\end{align}
\end{lemma}
\begin{proof}
See Appendix \ref{appendix_sens}.
\end{proof}
Lemma \ref{sens_lemma} states that under correct specification of the bias function (\ref{bias_function}), the conditional expectation of $H_{mk}^{a}(\gamma^{g*})-H_{mk-1}(\gamma^{g*})$ does not depend on $A_m$. This is the same crucial property satisfied by $H_{mk}(\gamma^{g*})-H_{mk-1}(\gamma^{g*})$ in (\ref{id}) that enabled identification of $\gamma^{g*}$. Thus, it follows that identification and estimation of $\gamma^{g*}$ under bias function (\ref{bias_function}) may proceed exactly as identification and estimation of $\gamma^{g*}$ under the parallel trends assumption (\ref{parallel_trends}) except substituting $H_{mk}^{a}(\gamma^g)-H_{m,k-1}(\gamma^{g})$ for $H_{mk}(\gamma^g)-H_{m,k-1}(\gamma^{g})$. Future work should explore suitable parameterizations of the bias function for tractable and informative sensitivity analysis.

\section{Relation to Other Work}\label{section_other_work}
\subsection{Estimands}
Most time-varying DiD work emphasizes effects of so-called `staggered adoption' strategies. In the staggered adoption setting, all units start at a common baseline level of treatment and then deviate from that baseline level at different times in a staggered fashion. In the simplest and most commonly considered case, the baseline treatment level is 0 or ‘untreated’, and units then initiate a binary treatment at different times. For times $k>m$, interest centers on $E[Y_k-Y_k(\bar{0})|\bar{A}_m=(\bar{0}_{m-1},1)]$, the effect of starting treatment at time $m$ compared to never starting treatment on the outcome at time $k$ among units that actually started treatment at time $m$. \citet{callaway2021difference} call these effects $ATT(m,k)$. We mentioned in Remark \ref{derived_estimands} that these estimands are identified in terms of SNMM parameters as $E[\gamma_{mk}^g(\bar{a}_m=(\bar{0}_{m-1},1),\bar{L}_m)|\bar{A}_m=(\bar{0}_{m-1},1)]$ for $g=\bar{0}$ and treatment coded as in Remark \ref{staggered}. Thus, SNMMs can target the same $ATT(m,k)$ estimands as standard DiD approaches and also model their heterogeneity as a function of time-varying covariates. 

Some work in the time-varying DiD literature has studied effects of interventions beyond initiation of a binary treatment. \citet{de2024difference}, for example, consider effects of initial changes in treatment level when not all units begin at the same baseline treatment and treatment may be non-binary. \citet{callaway2024difference} consider effects of initiation of continuous valued treatments from a shared baseline treatment. SNMMs of continuous valued treatments under parallel trends assumptions conditional on treatment history can straightforwardly model these effects as well. 

When treatment changes value repeatedly (e.g. switches on and off) in the data, $ATT(m,k)$ effects are marginal over post-initiation treatment patterns under the observational regime, similar to intention to treat effects in randomized trials with imperfect compliance. Just as per-protocol effects are often of interest in trials with non-compliance, effects of sustained interventions are often of interest in DiD settings. However, apart from \citet{renson2023identifying}, we are not aware of methods estimating these effects under parallel trends assumptions (see Section 3.4 of \citet{roth2023}). Effects of dynamic treatment strategies that may depend on covariates have also not been addressed under parallel trends except by \citet{renson2023identifying}.

In concurrent and cross-citing work, \citet{renson2023identifying} showed that a variation on the g-formula identifies $E[Y(g)]$ under parallel trends assumption (\ref{parallel_trends}). They proposed inverse probability weighted, outcome-regression based, and doubly robust estimators for this functional distinct from our own. The primary advantage offered by SNMMs compared to their estimators is the ability to model time-varying effect heterogeneity. SNMMs further enable simultaneous estimation of effects of a range of interventions. For example, in our crop yield application in Section \ref{crops}, we are able to jointly model the effects of a sustained change from any continuous valued GDD level to any other at any time assuming parallel trends (\ref{parallel_trends}) under the `keep treatment fixed' regime (\ref{continue_g}). \citet{renson2023identifying}, however, only provide methods to estimate the effect of a single GDD trajectory at a time. The primary drawback of SNMMs relative to the methods of \citet{renson2023identifying} is the requirement in practice to specify a parametric SNMM.

In other near concurrent work cross-citing our own, \cite{blackwell2024estimating} consider estimation of controlled direct effects under parallel trends. However, they consider a different data structure lacking an intermediate outcome measurement and accordingly require different assumptions.

Finally, our work is the only work we are aware of that leverages parallel trends assumptions to estimate optimal dynamic regimes as in Section \ref{section_general}. While we also require the extremely strong assumption of no effect modification by unobserved confounders, the method is one of the few \citep{zhang2024identification,han2021identification} to enable reinforcement learning in the presence of unobserved confounding.

\subsection{Parallel Trends Assumptions}
In the staggered adoption setting, the `never treated' parallel trends assumption from \citet{callaway2021difference} states that:
\begin{align}
\begin{split}\label{pt_cands}
     &\mathbb{E}\{Y_k(\bar{0})-Y_{k-1}(\bar{0})|\bar{L}_m,\bar{A}_{m-1}=\bar{0}_{m-1},A_m=1\}\\&=\mathbb{E}\{Y_k(\bar{0})-Y_{k-1}(\bar{0})|\bar{L}_m,\bar{A}_K=\bar{0}\}.
     \end{split}
\end{align}
(Here, we have transported their assumption to our notation. They denote their covariates by `$X$' and state that $X$ comprises `pre-treatment' information. They do not explicitly indicate with a time subscript that $X$ might evolve with time prior to treatment. However, they say that their framework allows for `potentially dynamic treatment selection' and their software allows users to input time-varying $X$ \citep{callaway2021getting}. We are interested in considering the setting where there are time-varying covariates, and the never-treated parallel trends assumption is made conditional on their evolving pre-treatment history.) In words, the `never treated' parallel trends assumption says that, given covariate history through time $m$, the expected counterfactual untreated future outcome trends in units that started treatment at time $m$ are equal to the expected outcome trends in units that in actuality never initiate treatment. This assumption is implausible if $\bar{L}$ contains a time-varying confounder $\bar Z$ of the trends. The issue is that the distribution of $\underline{Z}_{m+1}(\bar{0})$ in the never treated units comprising the conditioning event on the right hand side of (\ref{pt_cands}) would tend to differ from the distribution of $\underline{Z}_{m+1}(\bar{0})$ in the units comprising the conditioning event of the left hand side of (\ref{pt_cands}). This is because, by virtue of $\bar{Z}$ being a confounder, future values of the covariate $\underline{Z}_{m+1}(\bar{0})$ are associated with the future values of the treatment $\underline{A}_{m+1}$ that are conditioned on only in the never treated control group. As $\bar{Z}$ is a confounder of the trends, the trends will therefore also differ across the two groups, violating the assumption. See Figure \ref{fig:tree} in Appendix \ref{sec:appendix_tree} for a schematic illustrating this intuition.

It is straightforward to show that in the presence of a time-varying confounder $\bar Z$ contained in $\bar L$, whenever Assumption \ref{parallel_trends} holds, (\ref{pt_cands}) fails. Under Assumption \ref{parallel_trends}, the left hand side of (\ref{pt_cands}) is equal to $E[Y_k(\bar{0})-Y_{k-1}(\bar{0})|\bar{A}_m=0,\bar{L}_m]$. However, this quantity is not equal to the right hand side of (\ref{pt_cands}). This is because the conditioning event on the right hand side additionally contains $\underline{A}_{m+1}=0$, and $\underline{A}_{m+1}$ is not conditionally mean independent of $Y_k(\bar{0})-Y_{k-1}(\bar{0})$ given $(\bar{A}_m=0,\bar{L}_m)$ as it is confounded by $\underline{Z}_{m+1}$. The Single World Intervention Graph (SWIG) \citep{richardson2013single} in Figure \ref{fig:swig_sim} describes a simple data generating process with treatments at two time points in which Assumption \ref{parallel_trends} holds and (\ref{pt_cands}) fails, so long as it is the case that if conditional independence between two variables is not implied by the graph then they are conditionally mean dependent. The dependence of $Y_2(\bar{0})-Y_{1}(0)$ (denoted $\Delta_2(\bar{0})$ in the SWIG) on $A_{1}$ can be read off the SWIG. (We will later describe a simulation based on this SWIG.) 

Furthermore, if Assumption \ref{parallel_trends} fails, then (\ref{pt_cands}) will only hold if two loosely related quantities are coincidentally exactly equal. In particular, if Assumption \ref{parallel_trends} fails for some $\bar{l}_m$ and $k>m$, then let $c$ denote the magnitude of the violation, i.e. $c=E[Y_k(\bar{0})-Y_{k-1}(\bar{0})|\bar{A}_{m-1}=\bar{0},A_m=1,\bar{L}_m=\bar{l}_m]-E[Y_k(\bar{0})-Y_{k-1}(\bar{0})|\bar{A}_{m}=\bar{0},\bar{L}_m=\bar{l}_m]$. Let $c'$ denote the conditional mean association between $\underline{A}_{m+1}$ and $Y_k(\bar{0})-Y_{k-1}(\bar{0})$, i.e. $E[Y_k(\bar{0})-Y_{k-1}(\bar{0})|\bar{A}_{m}=\bar{0},\bar{L}_m=\bar{l}_m]-E[Y_k(\bar{0})-Y_{k-1}(\bar{0})|\bar{A}_{K}=\bar{0},\bar{L}_m=\bar{l}_m]$. For (\ref{pt_cands}) to hold, $c$ must equal $c'$. We cannot imagine any plausible mechanism linking and equalizing these largely unrelated quantities. Thus, we would consider it a striking coincidence if (\ref{pt_cands}) were to hold in the presence of a time-varying confounder. We note that the above discussion applies not just to the `never treated' parallel trends assumption but also to alternative assumptions that also condition on future treatment values, such as the `not yet treated' [by the time of the outcome] assumption \citep{callaway2021difference}.

In the setting where all covariates are time-invariant, however, Assumption \ref{parallel_trends} implies (\ref{pt_cands}) with $\bar L=Z_0$. The SWIG in Figure \ref{fig:swig_time_invariant} depicts conditions under which Assumption (\ref{parallel_trends}) holds in a two time point setting with only time-invariant confounding. It can be read off this SWIG that $A_1\independent Y_2(\bar{0})-Y_1(0)|Z_0,A_0=0$, where $Y_2(\bar{0})-Y_1(0)$ is again denoted by $\Delta_2(\bar{0})$ in the SWIG. Therefore, the right hand sides of the two parallel trends assumptions (which already share a left hand side) are equal, i.e. $E[Y_2(\bar{0})-Y_1(0)|Z_0,A_0=0]=E[Y_2(\bar{0})-Y_1(0)|Z_0,\bar{A}_1=0]$, implying that the assumptions are equivalent if this SWIG describes the data generating process. If Assumption \ref{parallel_trends} fails, then, as before, (\ref{pt_cands}) can only hold under a coincidental equality. 

\begin{figure}[h]
    \centering
    \includegraphics[scale=0.7]{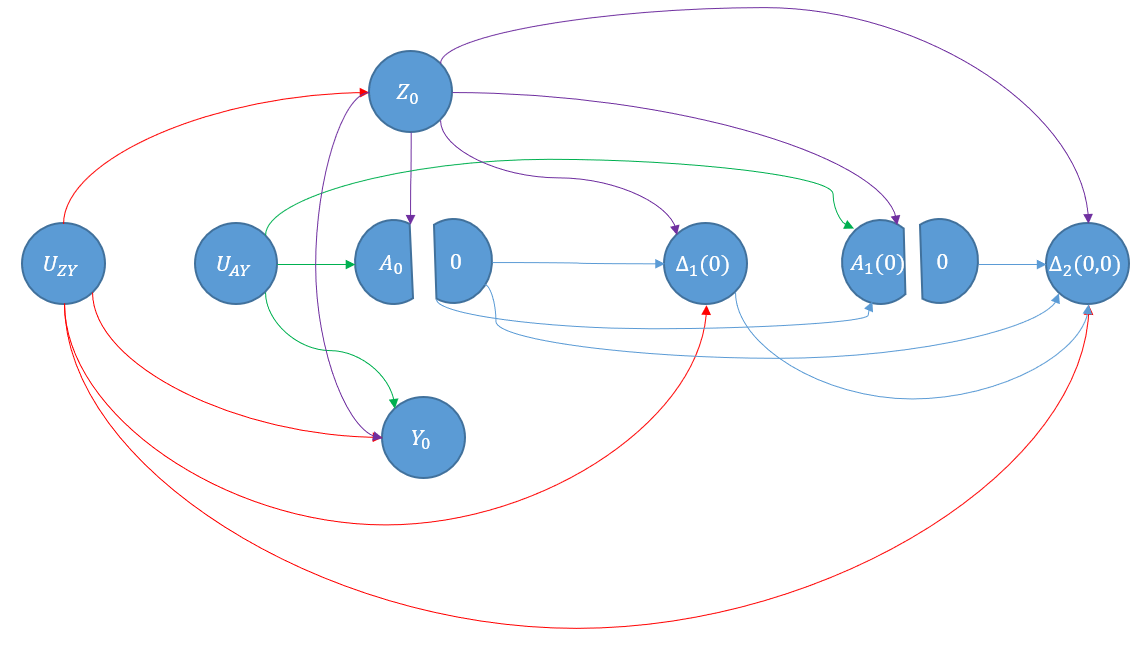}
    \caption{Single World Intervention Graph of a data generating process with only time-invariant confounding of trends in which Assumption \ref{parallel_trends} holds. The absence of arrows from $U_{AY}$ into the counterfactual increments $\Delta_t(\bar{0})$ is scale dependent. If the $Y_0\rightarrow A_1(0)$ is removed, then (\ref{pt_cands}) also holds as the SWIG implies that $A_1\independent \Delta_2(\bar{0})|A_0=0,L_0$, which means that conditioning on future treatment values does not break parallel trends.}
    \label{fig:swig_time_invariant}
\end{figure}

We illustrate aspects of the above discussion with a simple simulation in two time points. Consider the following data generating process (DGP):
\begin{align*}
    U\sim N(0,1); 
    Y_0\sim N(U,1); Z_0\sim Bernoulli(0.5); A_0\sim Bernoulli(logit^{-1}(Z_0+U))\\
    Y_1(0) \sim N(2Z_0+U); Z_1(0)\sim Bernoulli(logit^{-1}(-0.5+Z_0)); A_1\sim Bernoulli(logit^{-1}(Z_1(0)+U))\\
    Y_2(0,0)\sim N(Z_0+Z_1(0)+U,1).
\end{align*}
Let $\bar L_m = \bar Z_m$ at each time $m$. In this DGP, $\bar Z$ is an observed time-varying confounder of trends, i.e. it influences treatment and trends. The unobserved baseline confounder $U$ enters additively into the means of the counterfactual untreated outcomes. Thus, we might intuitively expect parallel trends assumptions to hold conditional on observed covariate history. The SWIG for this data generating process is depicted in Figure \ref{fig:swig_sim}. One can read off the SWIG, which incorporates the scale dependent fact that $U$ cancels out of the differences, that Assumption \ref{parallel_trends} holds. We generated 1,000,000 observations from this DGP (enough that sampling variability was negligible) and evaluated our parallel trends Assumption \ref{parallel_trends} and the `never treated' parallel trends assumption (\ref{pt_cands}) conditional on covariate history.

\begin{figure}[h]
    \centering
    \includegraphics[scale=0.75]{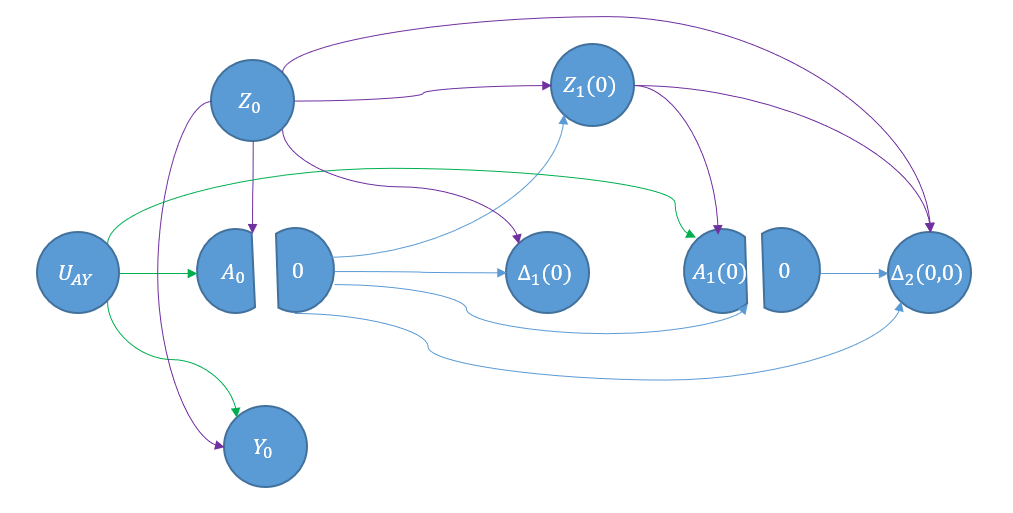}
    \caption{Single World Intervention Graph of the data generating process from our simulation scenario. The absence of arrows from $U_{AY}$ into the counterfactual increments $\Delta_t(\bar{0})$ is scale dependent and based on knowledge of the functional form of the data generating process. The conditional independencies that can be read off the SWIG imply that Assumption \ref{parallel_trends} holds. However, dependence of $A_1$ and $\Delta_2(\bar{0})$ via the $A_1(0)\leftarrow Z_1(0)\rightarrow \Delta_2(\bar{0})$ pathway implies that the never-treated parallel trends assumption (\ref{pt_cands}) would fail based on the argument in the main text.}
    \label{fig:swig_sim}
\end{figure}

Among those with $Z_0=0$, the trend among the never treated was $E\{Y_2(0,0)-Y_1(0)|A_0=A_1=0,Z_0=0\}=0.27$. The counterfactual untreated trend among those with $Z_0=0$ who were treated at time 0 was $E\{Y_2(0,0)-Y_1(0)|A_0=1,Z_0=0\}=0.37$. As these two trends were not equal, the `never treated' version of parallel trends is violated. The reason for the violation was that the average value of $Z_1(0)$ differed between the treated and control groups-- it was $0.27$ in the never treated group and $0.37$ in the $A_0=1$ group.  

By contrast, the counterfactual future untreated trend among those with $Z_0=0$ who were untreated at time 0 (but some of whom went on to be treated at time 1) was $E\{Y_2(0,0)-Y_1(0)|Z_0=0,A_0=0\}=0.37=\hat{E}\{Y_2(0,0)-Y_1(0)|Z_0=0,A_0=1\}$. Hence, our Assumption \ref{parallel_trends} did hold at level $L_0=Z_0=0$. The average value of $Z_1(0)$ was $0.37$ in both the $A_0=0$ and $A_0=1$ groups. A similar pattern occurred in the $Z_0=1$ group, with never treated parallel trends (\ref{pt_cands}) failing due to the $Z_1(0)$ variable's association with $A_1$ and Assumption \ref{parallel_trends} holding.  
R code for this simple illustration can be found at https://github.com/zshahn/marginal-structural-nested-mean-models-under-parallel-trends. 

\begin{figure}[h]
    \centering
    \includegraphics[scale=0.65]{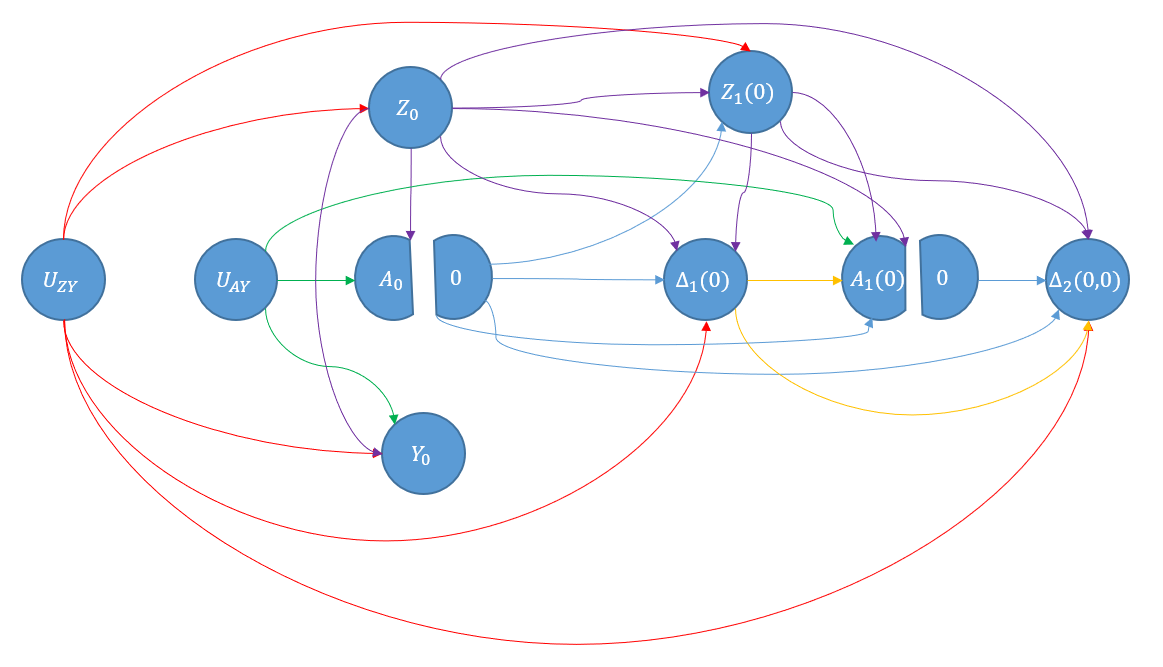}
    \caption{Single World Intervention Graph (SWIG) depicting a scenario under which our parallel trends assumption would hold. Key assumptions encoded by missing arrows include that future covariates are not confounded with treatment, and the baseline outcome does not cause future covariates or counterfactual increments.}
    \label{fig:swig}
\end{figure}

However, the question remains--under what conditions would Assumption \ref{parallel_trends} hold? Let $\bar L_m$ comprise $(\bar Z_m,\bar\Delta_m)$, i.e. the full covariate history and the history of outcome \textit{increments} but not outcome levels. Figure \ref{fig:swig} contains a SWIG depicting some key implications of Assumption \ref{parallel_trends} in the two time point setting. $\Delta_t(\bar{0})$ again denotes the difference $Y_t(\bar{0})-Y_{t-1}(\bar{0})$. Salient features of the graph include:
\begin{itemize}
    \item[(a)] Future covariates are unconfounded with past treatment, i.e. no arrows from $U_{AY}$ into $\bar{Z}$, no arrows from $U_{ZY}$ into $\bar{A}$, and independence of $U_{AY}$ and $U_{ZY}$.
    \item[(b)] The baseline outcome is `inert', i.e. there are no arrows from $Y_0$ into future outcomes, covariates, or $A_0$.
\end{itemize}
Condition (a) can be viewed as a generalization to the time-varying treatment setting of the result of \citet{caetano2022difference} showing that in the presence of time-varying confounding of a point exposure, parallel trends conditional on pre-treatment covariate values alone is possible if the post-treatment covariates are not confounded with treatment. Note also that $Y_0$ could not be included in the conditioning set $\bar L$, as doing so would open the collider path $A_1\leftarrow U_{AY}\rightarrow Y_0 \leftarrow U_{ZY}\rightarrow\Delta_2(0,0)$, thereby losing identification. Thus, under this SWIG, we cannot estimate effect heterogeneity as a function of past outcome values. We emphasize that Figure \ref{fig:swig} is scale dependent, and the structural requirements (e.g. constant additive association between $U_{AY}$ and the outcome over time) for the conditional indepedencies of counterfactual trends cannot be depicted in a SWIG. See \citet{ghanem2022selection} for a detailed discussion of structural conditions for parallel trends.

When we are interested in effects of interventions beyond initiation (e.g. blip effects, effects of sustained interventions, or CDEs) and the treatment might switch values multiple times in the data, then our parallel trends assumption (\ref{parallel_trends}) imposes additional restrictions compared to the staggered adoption setting when treatment is coded as in Remark \ref{staggered}. Future expected counterfactual trends are required to be equal in the treated and untreated given covariate and treatment history \textit{at all times}, not only when there has been no prior treatment. However, this more general assumption is a natural generalization of the staggered adoption setting, as `no prior treatment' is just a particular treatment history. The more general assumption is required for identification and estimation of effects of interventions involving treatments following initiation.

 While the interpretation of the parallel trends assumption (\ref{parallel_trends}) depends on $g$, it may be difficult to justify the assumption for one regime over another, especially if the regimes in question are complex. If a practitioner invokes (\ref{parallel_trends}) for an arbitrary complex regime of interest, a critic might argue that the practitioner is making the assumption out of convenience. That is, whatever treatment regime had been of interest, the analyst would have invoked the assumption corresponding to that regime and therefore effectively makes the assumption for all $g$. However, assuming parallel trends for \textit{all} regimes is a much stronger assumption than is usually made in the DiD literature and imposes restrictions on effect heterogeneity. In Appendix \ref{appendix_effect_mod}, we show that parallel trends under all regimes (practically) implies that there is no effect modification on the additive scale by unobserved confounders. 

\subsection{Placebo tests}
A common practice in DiD analyses is to examine ``event-study'' plots that include estimates of `placebo’ effects on outcomes at times prior to treatment initiation. Since future treatment cannot causally affect past outcomes, nonzero estimates for pre-treatment periods are evidence of confounding of treatment with past outcome increments. Confounding with past outcome increments, in turn, is taken as evidence of confounding with future increments, which would violate parallel trends. The ability to empirically check the plausibility of the parallel trends assumption in this way is a major selling point for DiD.

The typical rationale for expecting that future parallel trends should co-occur with pre-treatment parallel trends is an implicit fixed effects generative model for the outcome across time points. \citet{ghanem2026should} recently formalized the question of
when pre-treatment parallel trends should be expected to accompany the
post-treatment parallel trends assumptions used for identification in DiD.
They study this question in a three-period, two-group point-treatment setting
under a standard untreated-outcome two-way fixed effects model
\[
Y_{it}(0)=\alpha_i+\lambda_t+\varepsilon_{it}.
\]
They emphasize that this model represents a favorable benchmark for pre-trend
tests, since it encodes the idea that the untreated outcome structure is stable
across pre- and post-treatment periods. They do not consider the case of time-varying treatments.

We take a complementary graphical approach to reasoning about when pre-treatment placebo tests might be informative about our post-treatment parallel trends assumption (\ref{parallel_trends}) in settings with potentially time-varying treatments and covariates. In the absence of time-varying covariates, our reasoning leads to qualitatively similar conditions for informative pre-treatment placebo tests as \citep{ghanem2026should}, but without invoking a TWFE model.  

For simplicity, consider \(g=\bar 0\), the untreated regime. Suppose that the data
generating process is described by the SWIG in Figure \ref{fig:swig}, strengthened
so that pre-\(m\) untreated outcome increments do not cause subsequent treatments or covariates. That is, we assume that in the untreated SWIG: (a) there are no unmeasured common causes of treatments and increments $\bar\Delta(\bar 0)$; (b) there are no unmeasured common causes of treatments and covariates $\bar Z$; (c) past outcomes do not cause future increments, treatments, or covariates; and (d) past increments do not cause future treatments or covariates. Under
such a graph, $(\bar Z_m,\bar A_{m-1})$ is a sufficient adjustment set for Assumption
(\ref{parallel_trends}), so that past increments \(\bar\Delta_m\) are not included
in the adjustment set.

For \(q<m\), a natural placebo estimand compares
\[
E(\Delta_q \mid A_m=a_m,\bar A_{m-1}=\bar 0,\bar Z_m=\bar z_m)
\]
with
\[
E(\Delta_q \mid A_m=0,\bar A_{m-1}=\bar 0,\bar Z_m=\bar z_m).
\]
This is a pre-first-treatment diagnostic in that it is defined on the risk set
\(\bar A_{m-1}=\bar 0\). On this risk set, consistency implies
\[
\bar Z_m=\bar Z_m(\bar 0), \qquad \Delta_q=\Delta_q(\bar 0).
\]
Thus, if the SWIG implies
\[
A_m \perp\!\!\!\perp \Delta_q(\bar 0)
\mid \bar Z_m(\bar 0),\bar A_{m-1}=\bar 0,
\]
then it also implies the observed-data restriction
\[
E(\Delta_q \mid A_m,\bar A_{m-1}=\bar 0,\bar Z_m)
=
E(\Delta_q \mid \bar A_{m-1}=\bar 0,\bar Z_m).
\]
Hence, the population placebo contrast is zero, and the pre-treatment placebo provides a means to test a natural justification for parallel trends assumption (\ref{parallel_trends}). As with conventional event-study tests, failure is informative, but passing the test does not prove parallel trends---it only fails to reject one observable implication of the graphical sufficient condition.

There are notable limitations to such placebo testing. First, it is possible that parallel trends holds only conditional on past increments. However, a placebo test that conditioned on past increments would vacuously be 0 and therefore not be useful. The test is also limited to periods preceding initial treatment. Analogous diagnostics for later treatment decisions among units with 
nonzero prior treatment histories run into difficulties. Observed 
covariates through \(m\) may have been affected by earlier treatment, so 
\(\bar Z_m\) need not equal \(\bar Z_m(\bar 0)\), the variable appearing on the SWIG. Conditioning on the observed 
treatment-affected covariate history can induce associations between 
future treatment and untreated pre-\(m\) increments. For this reason, we view event-study diagnostics as primarily appropriate for initiation effects.

\section{\protect\normalsize Conclusion}\label{section_conclusion}

To summarize, we have shown that SNMMs expand the set of causal questions that can be addressed under
parallel trends assumptions. In particular, we have shown that additive and multiplicative general regime SNMMs are identified under time-varying conditional parallel trends
assumptions, which are more plausible in the presence of time-varying confounding than alternative parallel trends assumptions in the literature, and we have provided estimators. Using SNMMs, it is possible to do many things that were not
possible with standard DiD approaches, such as: characterize effect
heterogeneity as a function of time-varying covariates (e.g. unemployment in the Medicaid expansion analysis of Section \ref{medicaid}), estimate the effect
of one final blip of treatment (as in the flood analysis of Section \ref{floods}) and other derived contrasts, estimate controlled direct effects, estimate effects of sustained interventions when treatment changes value in the data (as in the crop yields example of Section \ref{crops}), and estimate time-varying effects of interventions on continuous valued treatments that do not share a baseline level (the crop example again). We have also explained how to estimate counterfactual expectations under a possibly complex dynamic regime as long as parallel trends holds with respect to that regime. The SWIG in Figure \ref{fig:swig} makes clear that the conditions required for parallel trends to hold in the presence of a time-varying confounder are rather stringent, but we have provided an approach for sensitivity analysis to violations. Under the stronger assumption that unobserved confounders are not effect modifiers on the relevant scale, we have further shown that optimal dynamic treatment regimes are identified via optimal regime SNMMs. We hope these new capabilities can be put to use in a wide variety of applications. 

\section*{Acknowledgements}
We would like to thank Andrea Rotnitzky for helpful discussions.

\bibliography{bibliography}

@article{laird1983further,
  title={Further comparative analyses of pretest-posttest research designs},
  author={Laird, Nan},
  journal={Am Stat},
  volume={37},
  number={4a},
  pages={329--330},
  year={1983},
  publisher={Taylor \& Francis}
}

@article{kennedy2023semiparametric,
  title={Semiparametric counterfactual density estimation},
  author={Kennedy, Edward H and Balakrishnan, Sivaraman and Wasserman, LA},
  journal={Biometrika},
  volume={110},
  number={4},
  pages={875--896},
  year={2023},
  publisher={Oxford University Press}
}

@article{kennedy2022semiparametric,
  title={Semiparametric doubly robust targeted double machine learning: a review},
  author={Kennedy, Edward H},
  journal={arXiv preprint arXiv:2203.06469},
  year={2022}
}

@book{chernozhukov2024applied,
  title={Applied Causal Inference Powered by ML and AI},
  author={Chernozhukov, Victor and Hansen, Christian and Kallus, Nathan and Spindler, Martin and Syrgkanis, Vasilis},
  year={2024}
}

@inproceedings{ghanem2026should,
  title={When Should Pre-trends Be Parallel?},
  author={Ghanem, Dalia and Sant’Anna, Pedro HC and W{\"u}thrich, Kaspar},
  booktitle={AEA Papers and Proceedings},
  volume={116},
  pages={64--69},
  year={2026},
  organization={American Economic Association 2014 Broadway, Suite 305, Nashville, TN 37203-2425}
}

@article{https://doi.org/10.3982/ECTA17522,
author = {Gerard J. van den Berg AND Johan Vikström},
title = {Long-Run Effects of Dynamically Assigned Treatments: a New Methodology and an Evaluation of Training Effects on Earnings},
journal = {Econometrica},
volume = {90},
number = {3},
pages = {1337-1354},
doi = {https://doi.org/10.3982/ECTA17522},
url = {https://onlinelibrary.wiley.com/doi/abs/10.3982/ECTA17522},
eprint = {https://onlinelibrary.wiley.com/doi/pdf/10.3982/ECTA17522},
year = {2022}
}

@article{lewis2020double,
  title={Double/debiased machine learning for dynamic treatment effects via g-estimation},
  author={Lewis, Greg and Syrgkanis, Vasilis},
  journal={arXiv preprint arXiv:2002.07285},
  year={2020}
}

@inproceedings{robins2004optimal,
  title={Optimal structural nested models for optimal sequential decisions},
  author={Robins, James M},
  booktitle={Proceedings of the Second Seattle Symposium in Biostatistics: analysis of correlated data},
  pages={189--326}, 
    year={2004},
  organization={Springer}
}

@article{newey1994large,
  title={Large sample estimation and hypothesis testing},
  author={Newey, Whitney K and McFadden, Daniel},
  journal={Handbook of econometrics},
  volume={4},
  pages={2111--2245},
  year={1994},
  publisher={Elsevier}
}

@article{jamie1994,
	title = {Correcting for non-compliance in randomized trials using structural nested mean models},
	volume = {23},
	number = {8},
	journal = {Communications in Statistics-Theory and Methods},
	author = {Robins, James M.},
	year = {1994},
	pages = {},
}

@incollection{jamie1997,
  title={ Causal inference from complex longitudinal data},
  author={Robins, James M.},
  booktitle={Latent
Variable Modeling and Applications to Causality, Lecture Notes in Statistics},
  pages={69--117},
  year={1997},
  publisher={Springer}
}

@article{snmm_review2014,
	title = {Structural Nested Models and G-estimation: the Partially Realized Promise},
	volume = {29},
	number = {4},
	journal = {Statistical Science},
	author = {Vansteelandt, Stijn and Joffe, Marshall},
	year = {2014},
	pages = {707-731},
}

@article{blackwell2024estimating,
  title={Estimating Controlled Direct Effects with Panel Data: An Application to Reducing Support for Discriminatory Policies},
  author={Blackwell, Matthew and Glynn, Adam and Hilbig, Hanno and Phillips, Connor Halloran},
  year={2024}
}

@inproceedings{de2024difference,
  title={Difference-in-Difference Estimators with Continuous Treatments and No Stayers},
  author={de Chaisemartin, Cl{\'e}ment and D'Haultf{\oe}uille, Xavier and Vazquez-Bare, Gonzalo},
  booktitle={AEA Papers and Proceedings},
  volume={114},
  pages={610--613},
  year={2024},
  organization={American Economic Association 2014 Broadway, Suite 305, Nashville, TN 37203}
}

@article{deschenes2012economic,
  title={The economic impacts of climate change: evidence from agricultural output and random fluctuations in weather: reply},
  author={Desch{\^e}nes, Olivier and Greenstone, Michael},
  journal={American Economic Review},
  volume={102},
  number={7},
  pages={3761--3773},
  year={2012},
  publisher={American Economic Association}
}

@article{gallagher2014learning,
  title={Learning about an infrequent event: Evidence from flood insurance take-up in the United States},
  author={Gallagher, Justin},
  journal={American Economic Journal: Applied Economics},
  pages={206--233},
  year={2014},
  publisher={JSTOR}
}

@article{robins1992identifiability,
  title={Identifiability and exchangeability for direct and indirect effects},
  author={Robins, James M and Greenland, Sander},
  journal={Epidemiology},
  volume={3},
  number={2},
  pages={143--155},
  year={1992},
  publisher={LWW}
}

@article{ciani2019dif,
  title={Dif-in-dif estimators of multiplicative treatment effects},
  author={Ciani, Emanuele and Fisher, Paul},
  journal={Journal of Econometric Methods},
  volume={8},
  number={1},
  pages={20160011},
  year={2019},
  publisher={De Gruyter}
}

@article{roth2023,
  title={What’s trending in difference-in-differences? A synthesis of the recent econometrics literature},
  author={Roth, Jonathan and Sant’Anna, Pedro HC and Bilinski, Alyssa and Poe, John},
  journal={Journal of Econometrics},
  year={2023},
  publisher={Elsevier}
}

@article{chaisemartin_review,
  title={Two-way fixed effects and differences-in-differences with heterogeneous treatment effects: A survey},
  author={De Chaisemartin, Cl{\'e}ment and d’Haultfoeuille, Xavier},
  journal={The Econometrics Journal},
  volume={26},
  number={3},
  pages={C1--C30},
  year={2023},
  publisher={Oxford University Press}
}

@article{renson2023identifying,
  title={Identifying and estimating effects of sustained interventions under parallel trends assumptions},
  author={Renson, Audrey and Hudgens, Michael G and Keil, Alexander P and Zivich, Paul N and Aiello, Allison E},
  journal={Biometrics},
  volume={79},
  number={4},
  pages={2998--3009},
  year={2023},
  publisher={Wiley Online Library}
}

@misc{chernozhukov2018double,
  title={Double/debiased machine learning for treatment and structural parameters},
  author={Chernozhukov, Victor and Chetverikov, Denis and Demirer, Mert and Duflo, Esther and Hansen, Christian and Newey, Whitney and Robins, James},
  year={2018},
  publisher={Oxford University Press Oxford, UK}
}

@article{smucler2019unifying,
  title={A unifying approach for doubly-robust $L_1$ regularized estimation of causal contrasts},
  author={Smucler, Ezequiel and Rotnitzky, Andrea and Robins, James M},
  journal={arXiv preprint arXiv:1904.03737},
  year={2019}
}

@article{robins1998correction,
  title={Correction for non-compliance in equivalence trials},
  author={Robins, James M},
  journal={Statistics in medicine},
  volume={17},
  number={3},
  pages={269--302},
  year={1998},
  publisher={Wiley Online Library}
}

@article{lok2012impact,
  title={Impact of time to start treatment following infection with application to initiating HAART in HIV-positive patients},
  author={Lok, Judith J and DeGruttola, Victor},
  journal={Biometrics},
  volume={68},
  number={3},
  pages={745--754},
  year={2012},
  publisher={Oxford University Press}
}

@techreport{callaway2024difference,
  title={Difference-in-differences with a continuous treatment},
  author={Callaway, Brantly and Goodman-Bacon, Andrew and Sant'Anna, Pedro HC},
  year={2024},
  institution={National Bureau of Economic Research}
}

@article{de2020two,
  title={Two-way fixed effects estimators with heterogeneous treatment effects},
  author={De Chaisemartin, Cl{\'e}ment and d’Haultfoeuille, Xavier},
  journal={American economic review},
  volume={110},
  number={9},
  pages={2964--2996},
  year={2020},
  publisher={American Economic Association 2014 Broadway, Suite 305, Nashville, TN 37203}
}

@incollection{robins2000sensitivity,
  title={Sensitivity analysis for selection bias and unmeasured confounding in missing data and causal inference models},
  author={Robins, James M and Rotnitzky, Andrea and Scharfstein, Daniel O},
  booktitle={Statistical models in epidemiology, the environment, and clinical trials},
  pages={1--94},
  year={2000},
  publisher={Springer}
}

@article{liu2021efficient,
  title={Efficient estimation of optimal regimes under a no direct effect assumption},
  author={Liu, Lin and Shahn, Zach and Robins, James M and Rotnitzky, Andrea},
  journal={Journal of the American Statistical Association},
  volume={116},
  number={533},
  pages={224--239},
  year={2021},
  publisher={Taylor \& Francis}
}

@article{rambachan2023more,
  title={A more credible approach to parallel trends},
  author={Rambachan, Ashesh and Roth, Jonathan},
  journal={Review of Economic Studies},
  volume={90},
  number={5},
  pages={2555--2591},
  year={2023},
  publisher={Oxford University Press US}
}

@inproceedings{de2024continuous,
  title={Difference-in-Difference Estimators with Continuous Treatments and No Stayers},
  author={de Chaisemartin, Cl{\'e}ment and D'Haultf{\oe}uille, Xavier and Vazquez-Bare, Gonzalo},
  booktitle={AEA Papers and Proceedings},
  volume={114},
  pages={610--613},
  year={2024},
  organization={American Economic Association 2014 Broadway, Suite 305, Nashville, TN 37203}
}

@article{zhang2024identification,
  title={On Identification of Dynamic Treatment Regimes with Proxies of Hidden Confounders},
  author={Zhang, Jeffrey and Tchetgen, Eric Tchetgen},
  journal={arXiv preprint arXiv:2402.14942},
  year={2024}
}

@article{caetano2022difference,
  title={Difference in differences with time-varying covariates},
  author={Caetano, Carolina and Callaway, Brantly and Payne, Stroud and Rodrigues, Hugo Sant'Anna},
  journal={arXiv preprint arXiv:2202.02903},
  year={2022}
}

@article{robins1986new,
  title={A new approach to causal inference in mortality studies with a sustained exposure period—application to control of the healthy worker survivor effect},
  author={Robins, James},
  journal={Mathematical modelling},
  volume={7},
  number={9-12},
  pages={1393--1512},
  year={1986},
  publisher={Elsevier}
}

@article{ghanem2022selection,
  title={Selection and parallel trends},
  author={Ghanem, Dalia and Sant'Anna, Pedro HC and W{\"u}thrich, Kaspar},
  journal={arXiv preprint arXiv:2203.09001},
  year={2022}
}

@misc{callaway2021getting,
  title={Getting started with the did package},
  author={Callaway, Brantly and Sant’Anna, Pedro},
  year={2021}
}

@article{richardson2013single,
  title={Single world intervention graphs (SWIGs): A unification of the counterfactual and graphical approaches to causality},
  author={Richardson, Thomas S and Robins, James M},
  journal={Center for the Statistics and the Social Sciences, University of Washington Series. Working Paper},
  volume={128},
  number={30},
  pages={2013},
  year={2013},
  publisher={Citeseer}
}

@article{han2021identification,
  title={Identification in nonparametric models for dynamic treatment effects},
  author={Han, Sukjin},
  journal={Journal of Econometrics},
  volume={225},
  number={2},
  pages={132--147},
  year={2021},
  publisher={Elsevier}
}

@article{schulte2014q,
  title={Q-and A-learning methods for estimating optimal dynamic treatment regimes},
  author={Schulte, Phillip J and Tsiatis, Anastasios A and Laber, Eric B and Davidian, Marie},
  journal={Statistical science: a review journal of the Institute of Mathematical Statistics},
  volume={29},
  number={4},
  pages={640},
  year={2014},
  publisher={NIH Public Access}
}

@article{callaway2021difference,
  title={Difference-in-differences with multiple time periods},
  author={Callaway, Brantly and Sant’Anna, Pedro HC},
  journal={Journal of econometrics},
  volume={225},
  number={2},
  pages={200--230},
  year={2021},
  publisher={Elsevier}
}

{\normalsize 
}

\section*{\protect\normalsize Appendix}

\section*{Appendix}
\appendix
\section{Proof of Theorem \ref{standard_id}}\label{appendix_theorem1}
Part (i)  
\begin{align}
\begin{split}
&E[U_{mk}(s_m,\gamma^*)] = E[E[U_{mk}(s_m,\gamma^*)|\bar{L}_m,\bar{A}_m]] \\
&=E[ \\
&E[H_{mk}(\gamma^{*})-H_{m,k-1}(\gamma^{*})|\bar{L}_m,\bar{A}%
_{m-1}](s_m(k,\overline{L}_{m},\overline{A}_{m})-E[s_m(k,\overline{L}_{m},%
\overline{A}_{m})|\bar{L}_m,\bar{A}_{m-1}]) \\
&] \\
&\text{by (\ref{id})} \\
&=E[ \\
&E[H_{mk}(\gamma^{*})-H_{m,k-1}(\gamma^{*})|\bar{L}_m,\bar{A}_{m-1}]\times
\\
&E[s_m(k,\overline{L}_{m},\overline{A}_{m})-E[s_m(k,\overline{L}_{m},%
\overline{A}_{m})|\bar{L}_m,\bar{A}_{m-1}]|\bar{A}_{m-1},\bar{L}_m] \\
&] \\
&\text{by nested expectations} \\
&=0
\end{split}%
\end{align}
The above establishes that the true blip functions are a solution to these
equations. The proof of uniqueness follows from the two Lemmas below. 

\begin{lemma}
{\normalsize \label{lemma1} Any functions $\gamma$ that satisfy (\ref%
{np_est_eq}) for all $s_m(k,\bar{l}_m,\bar{a}_m)$ also satisfy (\ref{id}),
i.e. 
\begin{equation*}
E[H_{mk}(\gamma) - H_{mk-1}(\gamma)|\bar{A}_m=\bar{a}_m,\bar{L}_m] =
E[H_{mk}(\gamma) - H_{mk-1}(\gamma)|\bar{A}_m=(\bar{a}_{m-1},g),\bar{L}_m]
\forall k>m.
\end{equation*}
}
\end{lemma}

\begin{proof}
{\normalsize Since (\ref{np_est_eq}) must hold for all $s_m(k,\bar{l}_m,\bar{%
a}_m)$, in particular it must hold for 
\begin{equation*}
s_m(k,\bar{l}_m,\bar{a}_m) = E[H_{mk}(\gamma) - H_{mk-1}(\gamma)|\bar{A}_m,%
\bar{L}_m].
\end{equation*}
Plugging this choice of $s_m(k,\bar{l}_m,\bar{a}_m)$ into $U_{mk}$, we get 
\begin{align*}
&E[\{H_{mk}(\gamma) - H_{mk-1}(\gamma)\}\times \\
&\{E[H_{mk}(\gamma) - H_{mk-1}(\gamma)|\bar{A}_m,\bar{L}_m]-E[H_{mk}(\gamma)
- H_{mk-1}(\gamma)|\bar{L}_m,\bar{A}_{m-1}]\}]=0 \\
&\implies \\
&E[H_{mk}(\gamma) - H_{mk-1}(\gamma)|\bar{A}_m,\bar{L}_m]=E[H_{mk}(\gamma) -
H_{mk-1}(\gamma)|\bar{L}_m,\bar{A}_{m-1}],
\end{align*}
which proves the result. }
\end{proof}

\begin{lemma}
{\normalsize \label{lemma2} $\gamma^*$ are the unique functions satisfying (%
\ref{id}) }
\end{lemma}

\begin{proof}
{\normalsize We proceed by induction. Suppose for any $k$ there exist other
functions $\gamma^{^{\prime }}$ in addition to $\gamma^*$ satisfying (\ref%
{id}) and satisfying $\gamma_{mk}^{^{\prime }}(\bar{l}_m,(\bar{a}_{m-1},g(\bar{l}_m,\bar{a}_{m-1}))=0
$. Then 
\begin{align*}
&E[H_{(k-1)k}(\gamma^*) - H_{(k-1)(k-1)}(\gamma^*)|\bar{A}_{k-1},\bar{L}%
_{k-1}] = E[H_{(k-1)k}(\gamma^*) - H_{(k-1)(k-1)}(\gamma^*)|(\bar{A}%
_{k-2},g),\bar{L}_{k-1}] \\
&E[H_{(k-1)k}(\gamma^{^{\prime }}) - H_{(k-1)(k-1)}(\gamma^{^{\prime }})|%
\bar{A}_{k-1},\bar{L}_{k-1}] = E[H_{(k-1)k}(\gamma^{^{\prime }}) -
H_{(k-1)(k-1)}(\gamma^{^{\prime }})|(\bar{A}_{k-2},g),\bar{L}_{k-1}].
\end{align*}
Differencing both sides of the above equations using the definition of $%
H_{mk}(\gamma)$ yields that 
\begin{align*}
&E[\gamma^*_{(k-1)k}(\bar{L}_{k-1},\bar{A}_{k-1}) - \gamma^{^{\prime
}}_{(k-1)k}(\bar{L}_{k-1},\bar{A}_{k-1})|\bar{A}_{k-1},\bar{L}_{k-1}] \\
&=\gamma^*_{(k-1)k}(\bar{L}_{k-1},\bar{A}_{k-1}) - \gamma^{^{\prime
}}_{(k-1)k}(\bar{L}_{k-1},\bar{A}_{k-1}) \\
&=E[\gamma^*_{(k-1)k}(\bar{L}_{k-1},\bar{A}_{k-1}) - \gamma^{^{\prime
}}_{(k-1)k}(\bar{L}_{k-1},\bar{A}_{k-1})|(\bar{A}_{k-2},g),\bar{L}_{k-1}] \\
&=0,
\end{align*}
where the last equality follows from the assumption that $%
\gamma_{mk}^{^{\prime }}(\bar{a}_m,\bar{l}_m)=0$ when $a_m=g(\bar{l}_m,\bar{a}_{m-1})$. This
establishes that $\gamma^*_{(k-1)k}=\gamma^{^{\prime }}_{(k-1)k}$ for all $k$%
. Suppose for the purposes of induction that we have established that $%
\gamma^*_{(k-j)k}=\gamma^{^{\prime }}_{(k-j)k}$ for all $k$ and all $j\in
\{1,\ldots,n\}$. Now consider $\gamma^*_{(k-(n+1))k}$ and $\gamma^{^{\prime
}}_{(k-(n+1))k}$. Again we have that 
\begin{align*}
&E[H_{(k-(n+1))k}(\gamma^*) - H_{(k-(n+1))(k-1)}(\gamma^*)|\bar{A}_{k-(n+1)},%
\bar{L}_{k-(n+1)}] = \\
&E[H_{(k-(n+1))k}(\gamma^*) - H_{(k-(n+1))(k-1)}(\gamma^*)|(\bar{A}%
_{k-(n+2)},0),\bar{L}_{k-(n+1)}] \\
&and \\
&E[H_{(k-(n+1))k}(\gamma^{^{\prime }}) - H_{(k-(n+1))(k-1)}(\gamma^{^{\prime
}})|\bar{A}_{k-(n+1)},\bar{L}_{k-(n+1)}] = \\
&E[H_{(k-(n+1))k}(\gamma^{^{\prime }}) - H_{(k-(n+1))(k-1)}(\gamma^{^{\prime
}})|(\bar{A}_{k-(n+2)},g),\bar{L}_{k-(n+1)}].
\end{align*}
And again we can difference both sides of these equations plugging in the
expanded definition of $H_{(k-(n+1))k}(\gamma)$ to obtain 
\begin{align}
\begin{split}  \label{lemma2_step1}
&E[\{Y_k - \sum_{j=k-(n+1)}^{k-1}\gamma^*_{jk}(\bar{A}_j,\bar{L}_j)\} -
\{Y_k - \sum_{j=k-(n+1)}^{k-1}\gamma^{^{\prime }}_{jk}(\bar{A}_j,\bar{L}%
_j)\}|\bar{A}_{k-(n+1)},\bar{L}_{k-(n+1)}] \\
&= E[\{Y_k - \sum_{j=k-(n+1)}^{k-1}\gamma^*_{jk}(\bar{A}_j,\bar{L}_j)\} -
\{Y_k - \sum_{j=k-(n+1)}^{k-1}\gamma^{^{\prime }}_{jk}(\bar{A}_j,\bar{L}%
_j)\}|(\bar{A}_{k-(n+2)},g),\bar{L}_{k-(n+1)}].
\end{split}%
\end{align}
Under our inductive assumption, (\ref{lemma2_step1}) reduces to 
\begin{align*}
&E[\gamma^*_{(k-(n+1))k}(\bar{L}_{k-(n+1)},\bar{A}_{k-(n+1)}) -
\gamma^{^{\prime }}_{(k-(n+1))k}(\bar{L}_{k-(n+1)},\bar{A}_{k-(n+1)})|\bar{A}%
_{k-(n+1)},\bar{L}_{k-(n+1)}] \\
&=\gamma^*_{(k-(n+1))k}(\bar{L}_{k-(n+1)},\bar{A}_{k-(n+1)}) -
\gamma^{^{\prime }}_{(k-(n+1))k}(\bar{L}_{k-(n+1)},\bar{A}_{k-(n+1)}) \\
&=E[\gamma^*_{(k-(n+1))k}(\bar{L}_{k-(n+1)},\bar{A}_{k-(n+1)}) -
\gamma^{^{\prime }}_{(k-(n+1))k}(\bar{L}_{k-1},\bar{A}_{k-(n+1)})|(\bar{A}%
_{k-(n+2)},g),\bar{L}_{k-(n+1)}] \\
&=0,
\end{align*}
proving that $\gamma_{(k-(n+1))k}^* = \gamma_{(k-(n+1))k}^{^{\prime }}$ for
all $k$. Hence, by induction, the result follows. }
\end{proof}

{\normalsize Uniqueness is a direct corollary of Lemmas \ref{lemma1} and \ref%
{lemma2}. }


{\normalsize Part (ii) 
\begin{align}
\begin{split}
&E[U^{\dagger}_{mk}(s_{m},\underline{\mathbf{\gamma}}_{m}^{*},\widetilde{v}%
_{m},\tilde{\pi}_m)] =
E[E[U^{\dagger}_{mk}(s_{m},\underline{\mathbf{\gamma}}_{m}^{*},\widetilde{v}%
_{m},\tilde{\pi}_m)|\bar{L}_m,%
\bar{A}_m]] \\
&=E[ \\
&E[H_{m,k}(\gamma^{*})-H_{m,k-1}(\gamma^{*})-\tilde{v}_m(k,\bar{L}_m,\bar{A}%
_{m-1},\gamma^{*})|\bar{L}_m,\bar{A}_{m-1}]q(\overline{L}_{m},\overline{A}%
_{m-1})\times \\
&(s_m(k,\overline{L}_{m},\overline{A}_{m})-E_{\tilde{\pi}_m}[s_m(k,\overline{L}_{m},\overline{A}_{m})|\bar{L}_m,\bar{A%
}_{m-1}]) \\
&] \\
&\text{by (\ref{id})} \\
&=E[ \\
&E[E[H_{m,k}(\gamma^{*})-H_{m,k-1}(\gamma^{*})-\tilde{v}_m(k,\bar{L}_m,\bar{A%
}_{m-1},\gamma^{*})|\bar{L}_m,\bar{A}_{m-1}]q(\overline{L}_{m},\overline{A}%
_{m-1})\times \\
&E[s_m(k,\overline{L}_{m},\overline{A}_{m})-\tilde{E}[s_m(k,\overline{L}_{m},\overline{A}_{m})|\bar{L}_m,%
\bar{A}_{m-1}]|\bar{L}_m,\bar{A}_{m-1}]] \\
&]
\end{split}%
\end{align}
Now the result follows because if $\tilde{v}_m(k,\bar{L}_m,\bar{A}%
_{m-1},\gamma^{*}) = v_m^{*}(k,\bar{L}_m,\bar{A}_{m-1},\gamma^{*})$, then 
\begin{equation*}
E[H_{mk}(\gamma^{*})-H_{m,k-1}(\gamma^{*})-\tilde{v}_m(k,\bar{L}_m,\bar{A%
}_{m-1},\gamma^{*})|\bar{L}_m,\bar{A}_{m-1}]= 0
\end{equation*}
and if $\tilde{E}[s_{mk}(\bar{L}_m,\bar{A}_m)|\bar{L}_m,\bar{A}_{m-1}]=E[s_{mk}(\bar{L}_m,\bar{A}_m)|\bar{L}_m,\bar{A}_{m-1}]$ then 
\begin{equation*}
E[s_m(k,\overline{L}_{m},\overline{A}_{m})-\tilde{E}[s_m(k,\overline{L}_{m},\overline{A}_{m})|\bar{L}_m,\bar{A%
}_{m-1}]|\bar{L}_m,\bar{A}_{m-1}]]= 0.
\end{equation*}
}

\section{Proof of Theorem \ref{theorem_est}}

\begin{proof}
Along the lines of Proposition 1 in \citet{kennedy2022semiparametric} and Lemma 3 in \citet{kennedy2023semiparametric}, we have that 
\begin{align*}
\sqrt{N}(\hat{\psi}^{cf}-\psi^{*})=&V(\psi^*,s,\pi^*)^{-1}\sqrt{N}(\mathbb{P}_N-\mathbb{P})\left\{\sum_{m=0}^KU_m(\psi^*,s_m,\pi^*,\nu^*)\right\}\\&+\sqrt{N}O_\mathbb{P}\left(\sum^{2}_{s=1}\left(\frac{n_s}{N}\right)\mathbb{P}\left\{\sum_{m=0}^K U_m(\psi^*,s_m,\hat{\pi}^{(-s)},\hat{\nu}^{(-s)})\right\}\right)+o_\mathbb{P}\left(1\right). 
\end{align*}
In order to apply Lemma 3, we require Assumptions 1, 2, 4 and 6. Then the main result follows so long as we can show that
\begin{align}\label{eq:key_res}
\sqrt{N}O_\mathbb{P}\left(\sum^{2}_{s=1}\left(\frac{n_s}{N}\right)\mathbb{P}\left\{\sum_{m=0}^K U_m(\psi^*,s_m,\hat{\pi}^{(-s)},\hat{\nu}^{(-s)})\right\}\right)=o_\mathbb{P}(1).
\end{align}

Fix $s=1$. We will proceed first by showing that each component of 
\begin{align}\label{eq:target_res}
\mathbb{P}\left\{\sum_{m=0}^K U_m(\psi^*,s_m,\hat{\pi}^{(-1)},\hat{\nu}^{(-1)})\right\}
\end{align}
is $o_\mathbb{P}(N^{-1/2})$. For any $m=0,...K$, 
\begin{align*}
&\sum^K_{m=0}U_m(O;\psi^*,s_m,\hat{\pi}^{(-1)},\hat{\nu}^{(-1)})\\&=\begin{pmatrix}
\sum^{K-1}_{m=0}\sum^{K}_{k=m+1}s_{mk1}(\bar{L}_m,\bar{A}_{m-1})\{A_m-\hat{\pi}^{(-1)}_m(\overline{L}_{m},%
\overline{A}_{m-1})\}\{H_{mk}(\psi^*)-H_{m,k-1}(\psi^*)-\hat{\nu}^{(-1)}_{mk}(\overline{L}_{m},%
\overline{A}_{m-1})\}\\
\vdots \\
\sum^{K-1}_{m=0}\sum^{K}_{k=m+1}s_{mkp}(\bar{L}_m,\bar{A}_{m-1})\{A_m-\hat{\pi}^{(-1)}_m(\overline{L}_{m},%
\overline{A}_{m-1})\}\{H_{mk}(\psi^*)-H_{m,k-1}(\psi^*)-\hat{\nu}^{(-1)}_{mk}(\overline{L}_{m},%
\overline{A}_{m-1})\}
\end{pmatrix}
\end{align*}
where $s_{mkj}(\bar{L}_m,\bar{A}_{m-1})$ is the  $(j,k)$th element of $s_m(\bar{L}_m,\bar{A}_{m-1})$. Then with some abuse of notation,
\begin{align*}
&\bigg|\mathbb{P}\left\{\sum^{K-1}_{m=0}\sum^{K}_{k=m+1}s_{mk1}(A_m-\hat{\pi}^{(-1)}_m)\{H_{mk}(\psi^*)-H_{m,k-1}(\psi^*)-\hat{\nu}^{(-1)}_{mk}\}\right\}\bigg|&\\
&= \bigg|\mathbb{P}\left\{\sum^{K-1}_{m=0}\sum^{K}_{k=m+1}s_{mk1}(\pi^*_m-\hat{\pi}^{(-1)}_m)(\nu_{mk}^*-\hat{\nu}^{(-1)}_{mk})\right\}\bigg|\\
&\leq \mathbb{P}\left\{\bigg|\sum^{K-1}_{m=0}\sum^{K}_{k=m+1}s_{mk1}(\pi^*_m-\hat{\pi}^{(-1)}_m)(\nu_{mk}^*-\hat{\nu}^{(-1)}_{mk})\bigg|\right\}\\
& \leq \mathbb{P}\left\{\sum^{K-1}_{m=0}\sum^{K}_{k=m+1}|s_{mk1}(\pi^*_m-\hat{\pi}^{(-1)}_m)(\nu_{mk}^*-\hat{\nu}^{(-1)}_{mk})| \right\}\\
&=\mathbb{P}\left\{\sum^{K-1}_{m=0}\sum^{K}_{k=m+1}|s_{mk1}||\pi^*_m-\hat{\pi}^{(-1)}_m||\nu_{mk}^*-\hat{\nu}^{(-1)}_{mk}|\right\} &\\
&\lesssim \mathbb{P}\left\{\sum^{K-1}_{m=0}\sum^{K}_{k=m+1}|\pi^*_m-\hat{\pi}^{(-1)}_m||\nu_{mk}^*-\hat{\nu}^{(-1)}_{mk}|\right\} &\\
&=O_{\mathbb{P}}\left(\sum_{m=0}^{K-1} \sum_{k=m+1}^{K}\norm{\hat{\pi}_{m} - \pi_{m}^*}\norm{\hat{\nu}_{mk}-\nu_{mk}^*}\right)\\
&=o_{\mathbb{P}}(N^{-1/2}) .
\end{align*}
where we repeatedly use the triangle inequality plus Assumptions 3 and 5. The same reason holds with $s=2$. Then because 
$n_1/N$ and $n_1/N$ both tend to a constant by assumption, the above results imply \eqref{eq:key_res}.
\end{proof}

\section{Relationship Between Universal Parallel Trends and No Additive Effect Modification by Unobserved Confounders Assumptions}\label{appendix_effect_mod}
To identify the counterfactual expectation $E[Y_k(g)]$ for a given $g$, we assumed parallel trends under $g$. One might argue, however, that if one assumes that parallel trends happen to hold for a particular regime of interest, one is really in effect assuming that parallel trends holds for \emph{all} regimes $g\in\mathcal{G}$. To identify optimal treatment strategies via optimal regime SNMMs, we needed to make the additional assumption (\ref{no_mod}) that there is no additive effect modification by unobserved confounders. Below, we sketch an argument that parallel trends under all regimes \emph{effectively}, though not strictly mathematically, implies no effect modification by unobserved confounders.

We will attempt to show the contrapositive of our result of interest, i.e. we will try to show that if there is additive effect modification by an unobserved confounder then parallel trends under all regimes cannot hold. Suppose that $\bar{L},\bar{U}$ is a minimal set satisfying sequential exchangeability. And suppose that $\bar{U}$ \emph{is} an effect modifier, i.e. for some $m$, $k$, and $g$
\begin{equation}\label{u_effect_mod}
    E[Y_k(\bar{a}_{m-1},\underline{g}_m)-Y_k(\bar{a}_{m-1},0,\underline{g}_{m+1})|\bar{A}_m,\bar{L}_m,\bar{U}_m=\bar{u}_m]
\end{equation}
depends on $\bar{u}_m$. If parallel trends does hold for all regimes, in particular it would hold for $\underline{g}_m$ and $(0,\underline{g}_{m+1})$, i.e.
\begin{equation}
    E[Y_k(\bar{a}_{m-1},\underline{g}_m)-Y_{k-1}(\bar{a}_{m-1},\underline{g}_m)|\bar{A}_m,\bar{L}_m]
\end{equation}
would not depend on $A_m$ and 
\begin{equation}
    E[Y_k(\bar{a}_{m-1},0,\underline{g}_{m+1})-Y_{k-1}(\bar{a}_{m-1},0,\underline{g}_{m+1})|\bar{A}_m,\bar{L}_m]
\end{equation}
would not depend on $A_m$. Therefore, the difference of the above two quantities also would not depend on $A_m$, i.e.
\begin{equation}\label{parallel_diff}
    E[Y_k(\bar{a}_{m-1},\underline{g}_m)-Y_k(\bar{a}_{m-1},0,\underline{g}_{m+1})|\bar{A}_m,\bar{L}_m]-E[Y_{k-1}(\bar{a}_{m-1},\underline{g}_m)-Y_{k-1}(\bar{a}_{m-1},0,\underline{g}_{m+1})|\bar{A}_m,\bar{L}_m]
\end{equation}
would not depend on $A_m$. But $U_m$ being an effect modifier (i.e. (\ref{u_effect_mod}) depending on $u_m$) and the assumption that $(\bar{L}_m,\bar{U}_m)$ is a minimal sufficient set together imply that the first conditional expectation in (\ref{parallel_diff}) \emph{does} depend on $A_m$. Now, it is possible for the full difference in (\ref{parallel_diff}) to still not depend on $A_m$ if the second conditional expection also depends on $A_m$ and in such a way as to cancel out the dependence on $A_m$ of the first conditional expectation. While this is technically possible, the type of cancelling out required is not plausible. So by reductio ad absurdity (though not the formal Latin absurdum indicating a true contradiction), parallel trends for all regimes ``effectively" implies no additive effect modification by unobserved confounders. 

\section{Proof of Lemma \ref{sens_lemma}, enabling sensitivity analysis}\label{appendix_sens}
Recall that we define bias function
\begin{align*}
\begin{split}\label{bias_function}
&c_{mk}^g(\bar{l}_m,
\bar{a}_m)=\\
&E[Y_k(\bar{a}_{m-1},\underline{g}_m)-Y_{k-1}(\bar{a}_{m-1},\underline{g}_m)|\bar{L}_m=\bar{l}_m,\bar{A}_m=\bar{a}_m] - \\
&E[Y_k(\bar{a}_{m-1},\underline{g}_m)-Y_{k-1}(\bar{a}_{m-1},\underline{g}_m)|\bar{L}_m=\bar{l}_m,\bar{A}_m=(\bar{a}_{m-1}, g(\bar{l}_m,\bar{a}_{m-1}))].
\end{split}
\end{align*}
We can then prove Lemma 1 as follows.
\begin{align*}
&E[H_{mk}^a(\gamma^*)-H_{m,k-1}(\gamma^*)|\bar{L}_m,\bar{A}_m]\\
&=E[H_{mk}(\gamma^*)-H_{m,k-1}(\gamma^*)-c_{mk}(\bar{A}_m,\bar{L}_m)|\bar{L}_m,\bar{A}_m]\\
&=E[Y_{k}(\bar{A}_{m-1},\underbar{g})-Y_{k-1}(\bar{A}_{m-1},\underbar{g})-c_{mk}(\bar{A}_m,\bar{L}_m)|\bar{L}_m,\bar{A}_m]\\
&=E[E[Y_k(\bar{A}_{m-1},\underbar{g})-Y_{k-1}(\bar{A}_{m-1},\underbar{g})|\bar{L}_m,\bar{A}_{m-1},A_m=g(\bar{L}_m,\bar{A}_{m-1})]|\bar{L}_m,\bar{A}_m]\\
&=E[Y_k(\bar{A}_{m-1},\underbar{g})-Y_{k-1}(\bar{A}_{m-1},\underbar{g})|\bar{L}_m,\bar{A}_{m-1},A_m=g(\bar{L}_m,\bar{A}_{m-1})]\\
&=E[H_{mk}(\gamma^*)-H_{m,k-1}(\gamma^*)-c_{mk}(\bar{A}_m,\bar{L}_m)|\bar{L}_m,\bar{A}_{m-1},A_m=g(\bar{L}_m,\bar{A}_{m-1})]\\
&=E[H_{mk}^a(\gamma^*)-H_{m,k-1}(\gamma^*)|\bar{L}_m,\bar{A}_{m-1},A_m=g(\bar{L}_m,\bar{A}_{m-1})]
\end{align*}
where the first equality is by the definition of $H_{mk}^a(\gamma^*)$, the second equality is by (\ref{H_cf}), the third equality is by the definition of $c_{mk}(\bar{A}_m,\bar{L}_m)$, the fourth equality is by nested expectations, the fifth equality is again by (\ref{H_cf}) and because $c_{mk}(\bar{A}_m,\bar{L}_m)=0$ when $A_m=g(\bar{L}_m,\bar{A}_{m-1})$, and the final equality is again because $c_{mk}(\bar{A}_m,\bar{L}_m)=0$ when $A_m=g(\bar{L}_m,\bar{A}_{m-1})$. This completes the proof.

\section{Schematic Illustration to Facilitate Comparison of Parallel Trends Assumptions}\label{sec:appendix_tree}

\begin{figure}[t]
\centering
\includegraphics[scale=.7]{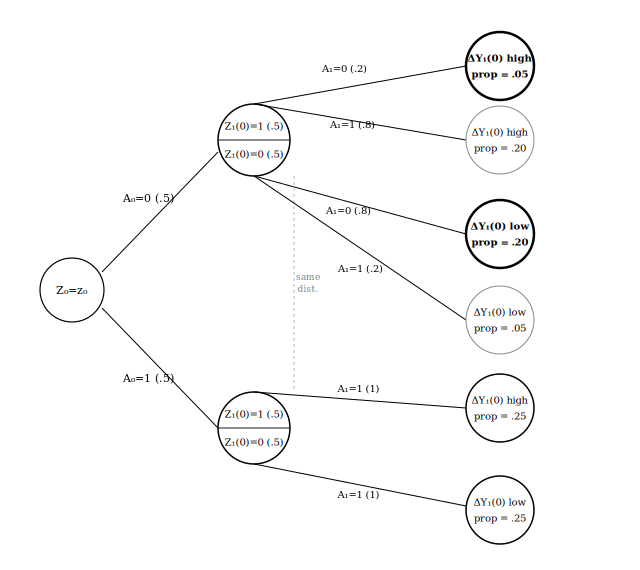}
\caption{
Tree illustrating a hypothetical distribution (conditional on baseline covariate $Z_0=z_0$) of treatments $A_0$ and $A_1$, time 1 untreated counterfactual covariate $Z_1(0)$, and counterfactual time 1 untreated trend $\Delta Y_1(0)$. The `prop' value in each leaf node states the proportion of the population in that leaf. The tree reflects a staggered adoption setting since $A_0=1$ implies $A_1=1$. The tree also reflects the assumption that $Z_1(0)$ is not confounded with $A_0$. In the depicted distribution, $A_1$ is more likely to be 1 when $Z_1=1$, and $\Delta Y_1(0)$ is more likely to be `high' when $Z_1(0)=1$. These properties of the distribution capture that $Z_1(0)$ is a time-varying confounder of the trend. Finally, the distribution depicts a scenario in which unobserved confounder $U$, not pictured, does not influence $\Delta Y_1(0)$.\\ 
\\
Our parallel trends assumption (\ref{parallel_trends}) says that the average $\Delta Y_1(0)$ value in the $A_0=1$ branch is equal to the average $\Delta Y_1(0)$ value in the $A_0=0$ branch. Examining the leaf proportions, we see that the $\Delta Y_1(0)$ values are an equal mix of `high' and `low' in both the $A_0=1$ and $A_0=0$ branches, satisfying our assumption. The Callaway and Sant'anna assumption (\ref{pt_cands}) states that the average $\Delta Y_1(0)$ value in the $A_0=1$ branch is equal to the average of the $\Delta Y_1(0)$ values in juts the $(A_0=0,A_1=0)$ nodes that are in bold. Because $Z_1(0)$ is disproportionately $0$ among those with $A_1=0$, the $\Delta Y_1(0)$ values in the $(A_0=0,A_1=0)$ nodes are disproportionately `low', which leads to a parallel trends violation. Their assumption could be rescued if $U$ impacted $\Delta Y_1(0)$ in such a way that the disparate $U$ distributions in the $A=0$ and $A=1$ branches exactly cancel the $Z_1(0)$-induced trend violation.      
}
\label{fig:tree}
\end{figure}

\end{document}